\documentclass{article}

\usepackage[utf8]{inputenc}

\usepackage{amsmath,amssymb,amsfonts,mathtools}
\usepackage{fullpage}
\usepackage{amsthm}
\usepackage{url}
\usepackage[colorlinks,citecolor=blue]{hyperref}
\usepackage{cleveref}
\usepackage[ruled,noend]{algorithm2e}
\usepackage{algpseudocode}
\SetKwInOut{Input}{Input}
\SetKwInOut{Output}{Output}
\usepackage{booktabs,adjustbox}
\usepackage[margin=1in]{geometry}

\newtheorem{theorem}{Theorem}
\newtheorem{definition}{Definition}

\newtheorem{lemma}{Lemma}
\newtheorem{observation}{Observation}
\newtheorem{proposition}{Proposition}
\newtheorem{claim}{Claim}

\DeclareMathOperator*{\E}{\mathbb{E}}
\DeclareMathOperator{\aut}{aut}
\DeclareMathOperator{\stab}{Stab}
\DeclareMathOperator{\orb}{Orb}
\newcommand{\set}[1]{\left\{#1\right\}}
\DeclarePairedDelimiter{\ceil}{\lceil}{\rceil}
\DeclarePairedDelimiter{\floor}{\lfloor}{\rfloor}

\newcommand{\Prs}[1]{\Pr\left[#1\right]}

\newcommand{\suchthat}{\,\middle\vert\,}

\usepackage{color}

\newcommand{\items}{\mathcal{M}}
\newcommand{\agents}{\mathcal{N}}
\newcommand{\alloc}{\mathcal{A}}

\newcommand{\sdpref}{\succeq^{\mathrm{sd}}}
\newcommand{\notsdpref}{\not\succeq^{\mathrm{sd}}}
\newcommand{\sdprefneq}{\succ^{\mathrm{sd}}}

\newcommand{\EF}{\mathrm{EF}}
\newcommand{\EFX}{\mathrm{EFX}}
\newcommand{\SDEF}{\mathrm{SD}\text{-}\mathrm{EF}}

\newcommand{\SDEFX}{\mathrm{SD}\text{-}\mathrm{EFX}}

\newcommand{\hatsig}{\widehat{\sigma}}
\newcommand{\hatA}{\widehat{A}}
\newcommand{\cE}{\mathcal{E}}

\newcommand{\hgraph}{H}
\newcommand{\hgraphset}{\mathcal{H}}

\usepackage{tikz}
\usepackage{ifthen}             
\usetikzlibrary{calc}           

\def\rotateclockwise#1{
  \newdimen\xrw
  \pgfextractx{\xrw}{#1}
  \newdimen\yrw
  \pgfextracty{\yrw}{#1}
  \pgfpoint{\yrw}{-\xrw}
}

\def\rotatecounterclockwise#1{
  \newdimen\xrcw
  \pgfextractx{\xrcw}{#1}
  \newdimen\yrcw
  \pgfextracty{\yrcw}{#1}
  \pgfpoint{-\yrcw}{\xrcw}
}

\def\outsidespacerpgfclockwise#1#2#3{
  \pgfpointscale{#3}{
    \rotateclockwise{
      \pgfpointnormalised{
        \pgfpointdiff{#1}{#2}}}}
}

\def\outsidespacerpgfcounterclockwise#1#2#3{
  \pgfpointscale{#3}{
    \rotatecounterclockwise{
      \pgfpointnormalised{
        \pgfpointdiff{#1}{#2}}}}
}

\def\outsidepgfclockwise#1#2#3{
  \pgfpointadd{#2}{\outsidespacerpgfclockwise{#1}{#2}{#3}}
}

\def\outsidepgfcounterclockwise#1#2#3{
  \pgfpointadd{#2}{\outsidespacerpgfcounterclockwise{#1}{#2}{#3}}
}

\def\outside#1#2#3{
  ($ (#2) ! #3 ! -90 : (#1) $)
}

\def\cornerpgf#1#2#3#4{
  \pgfextra{
    \pgfmathanglebetweenpoints{#2}{\outsidepgfcounterclockwise{#1}{#2}{#4}}
    \let\anglea\pgfmathresult
    \let\startangle\pgfmathresult

    \pgfmathanglebetweenpoints{#2}{\outsidepgfclockwise{#3}{#2}{#4}}
    \pgfmathparse{\pgfmathresult - \anglea}
    \pgfmathroundto{\pgfmathresult}
    \let\arcangle\pgfmathresult
    \ifthenelse{180=\arcangle \or 180<\arcangle}{
      \pgfmathparse{-360 + \arcangle}}{
      \pgfmathparse{\arcangle}}
    \let\deltaangle\pgfmathresult

    \newdimen\x
    \pgfextractx{\x}{\outsidepgfcounterclockwise{#1}{#2}{#4}}
    \newdimen\y
    \pgfextracty{\y}{\outsidepgfcounterclockwise{#1}{#2}{#4}}
  }
  -- (\x,\y) arc [start angle=\startangle, delta angle=\deltaangle, radius=#4]
}

\def\corner#1#2#3#4{
  \cornerpgf{\pgfpointanchor{#1}{center}}{\pgfpointanchor{#2}{center}}{\pgfpointanchor{#3}{center}}{#4}
}

\def\hedgeiii#1#2#3#4{
  \outside{#1}{#2}{#4} \corner{#1}{#2}{#3}{#4} \corner{#2}{#3}{#1}{#4} \corner{#3}{#1}{#2}{#4} -- cycle
}

\def\hedgem#1#2#3#4{
  
  \outside{#1}{#2}{#4}
  \pgfextra{
    \def\hgnodea{#1}
    \def\hgnodeb{#2}
  }
  foreach \c in {#3} {
    \corner{\hgnodea}{\hgnodeb}{\c}{#4}
    \pgfextra{
      \global\let\hgnodea\hgnodeb
      \global\let\hgnodeb\c
    }
  }
  \corner{\hgnodea}{\hgnodeb}{#1}{#4}
  \corner{\hgnodeb}{#1}{#2}{#4}
  -- cycle
}

\def\hedgeii#1#2#3{
  \hedgem{#1}{#2}{}{#3}
}

\usepackage{multirow}
\usepackage{wrapfig}
\usepackage[normalem]{ulem}

\allowdisplaybreaks

\title{On the Existence of Envy-Free Allocations \\ Beyond Additive Valuations}

 \author{
    Gerdus Benad\`e \\ 
    Boston University \\ 
    \texttt{\small{benade@bu.edu}}
    \and Daniel Halpern \\ 
    Harvard University \\ 
    \texttt{\small{dhalpern@g.harvard.edu}} 
    \and Alexandros Psomas \\ 
    Purdue University \\ 
    \texttt{\small{apsomas@cs.purdue.edu}} 
    \and Paritosh Verma \\ 
    Purdue University \\ 
    \texttt{\small{verma136@purdue.edu}}
}

\date{}

\begin{document}

\maketitle

\begin{abstract}

We study the problem of fairly allocating $m$ indivisible items among $n$ agents. Envy-free allocations, in which each agent prefers her bundle to the bundle of every other agent, need not exist in the worst case. However,   when agents have additive preferences and the value $v_{i,j}$ of agent $i$ for item $j$ is drawn independently from a distribution $D_i$,  envy-free allocations exist with high probability when $m \in \Omega( n \log n / \log \log n )$.

 In this paper, we study the existence of envy-free allocations under stochastic valuations far beyond the additive setting. We introduce a new stochastic model in which each agent's valuation is sampled by first fixing a worst-case function, and then drawing a uniformly random renaming of the items, independently for each agent. This strictly generalizes known settings;  for example, $v_{i,j} \sim D_i$ may be seen as picking a random (instead of a worst-case) additive function before renaming. We prove that random renaming is sufficient to ensure that envy-free allocations exist with high probability in very general settings. When valuations are non-negative and ``order-consistent,'' a valuation class that generalizes additive, budget-additive, unit-demand, and single-minded agents, SD-envy-free allocations (a stronger notion of fairness  than envy-freeness) exist for $m \in \omega(n^2)$ when $n$ divides $m$, and SD-EFX allocations exist for all $m \in \omega(n^2)$. The dependence on $n$ is tight, that is, for $m \in O(n^2)$ envy-free allocations don't exist with constant probability.
For the case of arbitrary valuations (allowing non-monotone,  negative, or mixed-manna valuations) and $n=2$ agents, we prove envy-free allocations exist with probability $1 - \Theta(1/m)$ (and this  is tight).

\end{abstract}

\newpage
\section{Introduction}

We consider the fundamental problem of fairly allocating a set $\items$ of $m$ indivisible items among a set $\agents$ of $n$ agents. Each agent $i$ has a valuation function $v_i : 2^\items \mapsto \mathbb{R}$, which maps each subset of items $S \subseteq \items$ to a value for $S$. In this domain, the gold standard of fairness is, arguably, \emph{envy-freeness}. An allocation $\alloc = (A_1, A_2, \ldots, A_n)$ is envy-free if each agent prefers her own bundle to the bundle of every other agent, that is, $v_i(A_i) \geq v_i(A_j)$ for all $i,j \in \agents$. It is easy to see that envy-free allocations do not exist in the worst case:    consider a single item and two agents valuing it positively. 

Motivated to circumvent this simple non-existence result, a line of research in fair division studies the existence of envy-free allocations under stochastic valuations. 
To date, this work has  focused on additive valuation functions.\footnote{A valuation function $v_i : 2^\items \mapsto \mathbb{R}$ is additive if $v_i(S \cup T) = v_i(S) + v_i(T)$ for all $S, T \subseteq \items$, $S \cap T = \emptyset$. An additive function $v_i$ can be succinctly represented with a value $v_{i,j}$ for each item $j \in \items$, such that $v_i(S) = \sum_{j \in S} v_{i,j}$ for all $S \subseteq \items$ and $i\in \agents$.}  
Dickerson et al.~\cite{dickerson2014computational} 
show that when agents have additive and non-negative valuation functions, and all item values are drawn independently from a distribution $\mathcal{D}$,  allocations that simultaneously satisfy envy-freeness and Pareto efficiency exist with high probability for $m \in \Omega(n \log n)$.\footnote{The result of~\cite{dickerson2014computational}, in fact, allows for limited correlation between the agents' values; see~\Cref{subsec:related} for details.} On the other hand, 
envy-free allocations do not exist with constant probability for $m \in n + o(n)$.  
In the same setting, Manurangsi and Suksompong~\cite{manurangsi2020envy} show that an envy-free allocation exists with high probability as long as $m\geq 2n$ and  $n$ divides $m$. When $m$ is not ``almost divisible'' by $n$,\footnote{Formally, when the remainder of the division is not between $n^\epsilon$ and $n - n^\epsilon$ for some constant $\epsilon \in (0,1)$.} an envy-free allocation is unlikely to exist for $m \in O(n \log n / \log \log n )$.
 Manurangsi and Suksompong~\cite{manurangsi2021closing} close this gap by proving that envy-free allocations exist with high probability when $m \in \Omega (n \log n / \log \log n )$. 
More recently, Bai and G{\"{o}}lz~\cite{bai2021envy} extend these bounds ($m \in \Omega (n \log n / \log \log n )$  for envy-freeness and $\Omega(n \log n)$   for envy-freeness plus Pareto efficiency) to the additive non-i.i.d.\ case where item values drawn independently from (agent specific) distributions $\mathcal{D}_i$.


Simply put, the goal of this paper is to study the existence of envy-free allocations under stochastic valuations beyond the additive case.

\subsection{Our Contribution}

In order to even pose the question of whether envy-free allocations exist beyond additive valuations, one (naturally) needs to first specify a stochastic model for generating such valuations. The aforementioned additive model is very natural: for every agent, we simply need to sample a value for each item, and additivity readily gives us the value for any subset of items we want. How should we go about sampling a, e.g., submodular valuation function?

The Bayesian setting is the dominant paradigm in mechanism design, where the worst-case lens fails to provide any useful insights or meaningful guarantees for various fundamental problems, e.g. revenue maximization. One standard ``beyond additive''  model in Bayesian mechanism design, e.g. in the literature on prophet inequalities (see~\cite{lucier2017economic} for a survey), assumes that valuation functions are drawn from distributions over families of functions (e.g. submodular or XOS functions). Another standard model, e.g. in auctions, is the ``$\mathcal{C}$ over independent items'' model for a condition $\mathcal{C}$ (e.g. ``subadditive over independent items''), introduced by~\cite{rubinstein2018simple}, where, informally, an agent's valuation function is parameterized by a vector of types, drawn from a product distribution.\footnote{See~\cite{rubinstein2018simple,rubinstein2017combinatorial,cai2017simple,chawla2016mechanism} for more details  about this model.} Unfortunately, in both these models it is easy to pick a valuation function/distribution over valuation functions such that item $1$ has more value than all other items combined, in which case envy-free allocations trivially don't exist; we need a new approach. Our first insight is that, to bypass such trivial lower bounds, we need a stochastic model that is \emph{neutral} with respect to items, that is, there is no a priori ``discrimination'' between items.

\paragraph{Our model.} 
Our first contribution is to introduce such a neutral model for stochastic valuations. First, fix a worst-case valuation function for each agent. Then, rename the items uniformly at random and independently across agents.  
Slightly more formally,  fix a worst-case valuation function $v_i: 2^{|\items|}\rightarrow \mathbb{R}$ for agent $i$, and sample a uniformly random permutation $\pi_i: \items \rightarrow \items$. The valuation $v_i^{\pi_i}(S)$ for a subset of items $S$ after renaming is equal to $v_i( \pi_i^{-1}(S) )$.

Observe that this simple model generalizes the standard stochastic additive setting where $v_{i,j} \sim \mathcal{D}_i$, that is, positive results in our setting imply positive results in the old setting. To see why this is the case, observe that the distribution over values is invariant with respect to taking a random permutation of the items (renaming). Therefore, sampling a random additive function (instead of a worst-case additive function) and then renaming the items at random, is equivalent to sampling according to the old setting. 

Initially, it might appear that our model should still allow trivial lower bounds, since renaming does not tell us anything about how \emph{values} for bundles concentrate. This is crucial, since concentration of values was necessary for the analysis of all (to the best of our knowledge) previous ``envy-free with high probability'' results. So, consider, for example, an identical additive function ($v_{i,j} = 1$ for all $i \in \agents, j \in \items$), where random renaming has no power. For this function, a necessary condition for envy-free allocations to exist is that $m$ is divisible by $n$. For an arbitrary worst-case function, one may naturally expect additional conditions (beyond divisibility) for envy-free allocations to exist with high probability. Surprisingly,  this is not the case: divisibility is \emph{sufficient} for strong positive results in our model. 

\paragraph{EF allocations under arbitrary valuations.} In \Cref{sec: general}, we stress-test our model. We study \emph{arbitrarily general} valuations for the case of $n=2$ agents. We impose no constraint on the valuation function: it can be superadditive, non-monotone, or negative for some bundles and positive for others. By picking a valuation function such that bundles of size strictly less than $m/2$ and strictly more than $m/2$ are worthless and noticing that random renaming will not affect this property, we conclude that, if envy-free allocations exist, they must allocate exactly $m/2$ items to each agent (and therefore, $m$ must be even). Surprisingly, in \Cref{theorem:ef-whp-gen} we prove that, for an even number of items and two agents with arbitrary (!) valuation functions, envy-free allocations exist with high probability. Specifically, the probability that an envy-free allocation exists after random renaming is $1 - \frac{1}{m/2+1}$; this bound is almost tight, as there exist instances with additive valuations such that an envy-free allocation doesn't exist with probability $1/m$.


The proof of~\Cref{theorem:ef-whp-gen} is based on the following insight: a valuation function can be represented as an $m/2$-uniform hypergraph on $m$ vertices, where a hyperedge of size $m/2$ indicates that the corresponding subset of items is preferred to its complement. In this representation, a random permutation of items' names corresponds to picking a random, isomorphic hypergraph. Whenever the hypergraphs corresponding to the two agent's valuation functions are not identical an envy-free allocation exists --- allocate $i$ the bundle corresponding to the hyperedge present in $i$'s hypergraph and not $j$'s.
Using the orbit-stabilizer theorem, we reduce our question about the existence of envy-free allocations to a question about the number of  automorphisms  of  $k$-uniform hypergraphs. In \Cref{theorem:auto-bound}, which may be of independent interest, we give the main technical ingredient needed to establish \Cref{theorem:ef-whp-gen}: we prove that the number of automorphisms of a $k$-uniform hypergraph on $m$ vertices is at most $\frac{m!}{m-k+1}$.

\paragraph{EF allocations under order-consistent valuations.}
In \Cref{sec:order} we consider a structured valuation class that strictly generalizes additive valuations. Given an order over the items $\pi$, we say that a subset of items $A$ stochastically dominates a subset of items $B$, denoted as $A \sdpref_\pi B$, when the best item in $A$ (according to $\pi$) is weakly better than the best item in $B$,  the second best item in $A$ is weakly better than the second best item in $B$, and so on.
We say that a valuation function $v$ is \emph{order-consistent} with respect to $\pi$ if, for all bundles $A, B \subseteq \items$,  $A \sdpref_\pi B$ implies that $v(A) \geq v(B)$. By picking $\pi$ to be the items sorted in order of decreasing  value, it is clear that additive valuations are order-consistent. Similarly, budget additive,\footnote{A valuation is budget additive with budget $B$ if each item $j$ has value $v_{j}$ and the value of   $S\subseteq \items$ is $v(S) = \min\{B, \sum_{j\in S} v_j\}$. 
} single-minded,\footnote{A valuation is single-minded if there is a subset $S^*$, such that $v(S) = v(S^*) > 0$ for all $S \supseteq S^*$, and $v(S) = 0$ otherwise.} and unit-demand\footnote{A valuation is unit-demand if there is a value $v_j$ for each item $j$, and the value for a subset of items $S$ is equal to $\max_{j \in S} v_j$.} valuations are also order-consistent. Even though seemingly restrictive, the class of order-consistent valuation functions is incomparable with large valuation classes, e.g., subadditive valuations (that is, there exist order-consistent valuations that are not subadditive).
The class of order-consistent valuations was also considered by Bouveret et al.~\cite{bouveret2010fair}, who study various algorithmic and complexity questions (in the worst-case model).

We prove that, given $n$ agents with arbitrary order-consistent valuation functions over $m$ items, where $m$ is divisible by $n$, the probability that an envy-free allocation exists after random renaming is at least $1 - O\left(\frac{n^2}{m} + \frac{n \log m}{m^{\frac{n - 1}{n}}} \right)$ (\Cref{thm:positive for order consistent}). In fact, we prove the existence of a (much) stronger notion of fairness, SD-envy-freeness~\cite{bogomolnaia2001new}: an allocation $\alloc$ is SD-envy-free if, for all agents $i,j \in \agents$, $A_i \sdpref_{\pi_i} A_j$. Note that, if an allocation is SD-envy-free, then it is envy-free for all additive utility functions consistent with the agents' preferences~\cite{aziz2015fair}. As a corollary to~\Cref{thm:positive for order consistent}, we get that for $m \in \omega(n^2)$, SD-envy-free allocations exist with high probability. By ``high probability'' we mean that for all $n$ and $\delta > 0$, there exists a $m_0 = m_0(n,\delta)$ such that, for all $m > m_0$, the probability is at least $1-\delta$. 

Our proof of~\Cref{thm:positive for order consistent} is constructive; we show that a simple Round-Robin process (agents take turns picking the best, according to their order, available item) produces such an allocation. Specifically, for arbitrary agents $i, j\in \agents$, we upper bound the probability that $i$ does not sd-prefer their bundle over $j$'s, i.e.\ $A_i \notsdpref_{\pi_i} A_j$. For this to happen, there must be some $1\leq k\leq m/n$ such that $i$ prefers $j$'s $k$-th best item (according to $i$) over $i$'s $k$-th best item. At a high level, we'd like to compute this probability, and then take a union bound over $k$ (and then another union bound over all pairs of agents). Notice that for $k=1$ this probability is already $\Theta(n/m)$, so, in order to afford all the union bounds, it better be the case that the true probabilities of the bad events are much smaller than $n/m$, and that our analysis is relatively tight. For all items $k < m/n$, i.e.\ all items except the last one picked by $i$, we can directly upper bound the probability that $i$ prefers $j$'s $k$-th best item. The analysis leverages the insight that, from $i$'s perspective, and over the random draws of $\pi_{-i}$, items picked by other agents look like (uniformly) random selections from the pool of remaining items; therefore, the distribution of other agents' bundles is identical. For $i$ to prefer $j$'s $k$-th best item over their own, $j$ must have selected $k$ items all better than $i$'s $k$-th pick, which can only occur if these items have all been picked in rounds $1$ through $k$; this event is unlikely. The precise bound is $\Theta(1/{\binom{m/n}{k}})$ (see Lemma~\ref{lem: bound for k < q, EF}), and requires carefully accounting for the items remaining at each step of Round-Robin, coupled with careful applications of known facts about the gamma function (e.g. Gautchi's inequality~\cite{gautschi1959some}). When $k=m/n$, previous arguments fail to yield a sufficiently small probability. Instead, define $L$ as the set of (roughly) $3n\log{m}$ worst items of agent $i$. We show that neutrality implies that, with high probability, by the time $i$ picks their $m/n$-th and worst item, all items in $L$ have been picked by others. When this happens, $i$ must like their worst item more than the worst (from $i$'s perspective) item of every agent who picked an item in $L$. Moreover, conditioned on all items in $L$ being picked by the time of $i$'s last pick, it is very likely that every other agent received at least one of these items in $L$.

Since $|A_i| > |A_j|$ immediately implies that $A_j \not\succeq_{\pi_j} A_i$,   $m$ being divisible by $n$ is a necessary condition for envy-freeness. Moreover,  standard birthday paradox arguments imply   (even for unit-demand valuations) that envy-free allocations may not exist with constant probability  for $m \in O(n^2)$. Therefore, \Cref{thm:positive for order consistent} is tight both in terms of the divisibility assumption and the bound on $m$.

When $n$ does not divide $m$,  we guarantee a notion weaker than $\SDEF$, called $\SDEFX$. In an $\EFX$ allocation, it holds that every agent $i$ does not envy a different agent $j$ after the removal of any item from $j$'s bundle; the definition of $\SDEFX$ (with respect to EFX) is analogous to the definition of $\SDEF$ (with respect to EF). $\EFX$ is considered the ``best fairness analog of envy-freeness'' in discrete fair division~\cite{caragiannis2019envy}. The existence of $\EFX$ allocations  remains an elusive open problem. In contrast, $\SDEFX$ allocations (which is a strictly stronger notion of fairness) do not exist in the worst-case, even for additive valuations (see Appendix \ref{appendix:no-sdefx}). We prove that, for the more general class of order-consistent valuations, $\SDEFX$ allocations exist with probability at least $1 - O\left(\frac{n^2}{m} + \frac{n \log m}{m^{\frac{n - 1}{n}}} \right)$.

Though the relation between $m$ and $n$ is asymptotically tight in~\Cref{thm:positive for order consistent}, we show that it is possible to get better probability bounds for the important case when $n$ is small, i.e., when $n$ is a constant. Taking, for example, the case of $n=2$, Round-Robin will not find an $\SDEF$ allocation with probability $\Omega(1/\sqrt{m})$, while a careful analysis would say that such allocations don't exist with probability $\Omega(1/m)$. This difference is caused by the fact that if $n$ is small, the last item in Round-Robin is given to the  ``wrong'' person. This motivates a new algorithm to close this gap. Our algorithm,  ``Give-Away Round-Robin,'' initially has every agent give each other agent their least desired remaining item and then proceeds with  the standard Round-Robin algorithm.
The probability of Round-Robin failing is greatest at the beginning and end of the process: at the beginning of the process, $\SDEF$ fails if any other agent picks $i$'s most preferred item; towards the end, there is a risk that all the remaining items are bad. 
Give-Away Round Robin gains by ensuring that every other agent receives one of $i$'s worst items, making sure that agent $i$ does not receive her least preferred items. At the same time, the analysis becomes significantly trickier. 
For Round-Robin we could be certain that, other than the last item, items picked by agent $i$ are ``good'' in the sense that they were the best in some pool of items. For Give-Away Round-Robin, each agent essentially starts with a small, random set of items. Therefore, arguing that SD-EF does not fail because of, say,  the first ten items is not clear at all. In our analysis, we consider various cases (failure because of ``high,'' ``middle,'' ``low,'' and last items), which need delicate, separate handling. In~\Cref{thm:give-away}, we show that, when $m$ is divisible by $n$, Give-Away Round-Robin outputs an $\SDEF$ allocation with probability at least $1 - \tilde{O}_n\left(\frac{1}{m} \right)$.

\subsection{Related Work}\label{subsec:related}

Dickerson et al.~\cite{dickerson2014computational} initiated the study of asymptotic fair division and showed that the welfare maximizing algorithm (allocate each good to the agent with the highest value for it) is envy-free with high probability for $m \in \Omega(n \log n)$ for the case of additive, non-negative valuations. This result holds when items' values are drawn i.i.d. from a common prior, but also when agents are correlated, but, for every item $j$: (1) the probability that agent $i$ has the highest value for $j$ is exactly $1/n$, and (2) the expected value of $v_{i,j}$ conditioned on $i$ having the highest value for $j$, is bigger, by constant, than the expected value of $v_{i,j}$ conditioned on some other agent having the highest value for $j$. Manurangsi and Suksompong~\cite{manurangsi2020envy,manurangsi2021closing} establish tight bounds for the existence of envy-free allocations in the i.i.d. model: $m \in \Omega (n \log n / \log \log n )$ is a necessary and sufficient condition; similar to our~\Cref{thm:positive for order consistent}, this bound is achieved by the classic Round-Robin algorithm.  Bai and G{\"{o}}lz~\cite{bai2021envy} extend these results to the case of independent but non-identical additive agents. 
Kurokawa et al.~\cite{kurokawa2016can}, Amanatidis et al.~\cite{amanatidis2017approximation}, and Farhadi et al.~\cite{farhadi2019fair} show that weaker notions of fairness, namely maximin share fairness, also exist with high probability. Farhadi et al.~\cite{farhadi2019fair} also study a ``stochastic items'' model, where every item $j$ has a probability distribution $\mathcal{D}_j$ (and $v_{i,j}$ is drawn from $\mathcal{D}_j$); interestingly, this model is not neutral. Suksompong~\cite{suksompong2016asymptotic} shows that proportional allocations exist with high probability, while Manurangsi and Suksompong~\cite{manurangsi2017asymptotic} study envy-freeness when items are allocated to groups rather than to individuals. For non-additive stochastic valuations, Manurangsi and Suksompong~\cite{manurangsi2021closing} and Gan et al.~\cite{gan2019envy} study the so-called \emph{house allocation} problem, where there are $m$ houses, $n$ agents, with $v_{i,j} \sim D_i$ for each agent $i$ and house $j$, and each agent must be allocated a single house; equivalently, one can think of unit-demand valuations and independently picking a uniformly random ranking of the items for each agent. The former paper shows that an envy-free assignment is likely to exist if $m/n > e + \epsilon$, for any constant $\epsilon > 0$, but unlikely to exist if $m/n < e - \epsilon$.

The existence of fair allocations, and their compatibility with efficiency, has also been studied in dynamic settings with additive valuations \cite{benade2018make}. The algorithms of Dickerson et al.~\cite{dickerson2014computational} and Bai and G{\"{o}}lz~\cite{bai2021envy} can be readily applied even when items arrive in an online fashion, implying that Pareto efficiency ex-post and envy-freeness with high probability are compatible even in an online setting. The same result was recently shown to be (approximately) true even when the designer doesn't have access to the exact values of the agents~\cite{benade2022dynamic}, or even when agents' valuations are correlated~\cite{zeng2020fairness,benade2022fair}.

Bouveret et al.~\cite{bouveret2010fair}  study preferences represented by SCI-nets, which take  the form of a partial order over bundles induced by a linear ordering over items. Order-consistent valuations are exactly those where pairwise dominance of bundles implies envy-freeness. As opposed to our interest here,~\cite{bouveret2010fair} are interested in computational aspects of the problem (e.g., the complexity of computing an envy-free allocation in the worst-case).


A related research direction is smoothed analysis. Traditionally, in smooth analysis, an instance is generated by starting from a worst-case instance and adding a small amount of noise. This model has been widely successful in circumventing computational hardness~\cite{spielman2004smoothed,blum1995coloring,feige1998heuristics,kolla2011play}. More recently, this model has been used to bypass impossibility results in social choice~\cite{xia2020smoothed,xia2021semi,flanigan2022smoothed,xia2021semi}, mechanism design~\cite{psomas2019smoothed,blum2021incentive,blum2017opting}, and, closer to this paper, fair division~\cite{bai2022fair}. Specifically, Bai et al.~\cite{bai2022fair} start from a worst-case instance for additive agents and add an independent boost to the value of each item for each agent; the authors give sufficient conditions (on the boosting) for envy-free allocations to exist with high probability. 
One could   interpret our model as a smoothed/semi-random model  which starts from a worst-case instance and adds noise in the form of random renaming. In contrast with existing models, we permute the items and leave the valuation function untouched. 


\section{Preliminaries}

We consider the problem of dividing a set $\items$ of $m$ indivisible items among a set $\agents$ of $n \ge 2$ agents. Throughout, we assume that $\agents = \{1,2, \ldots, n\}$ and the items are indexed by $1, 2, \ldots, m$. 
 An allocation $\alloc = (A_1, A_2, \ldots, A_n)$ is an $n$-partition of $\items$, where $A_i$ denotes the set of items allocated to agent $i \in \agents$. That is, in any allocation $\alloc$, all items are allocated,\footnote{Without this restriction, envy-free allocations trivially exist, since we can simply not allocate any items.} $\bigcup_{i \in \agents} A_i = \items$, and each item is allocated to only one agent, $A_i \cap A_j = \emptyset$ for all distinct agents $i, j \in \agents$. We sometimes refer to a subset of items as a  \emph{bundle}.

Each agent $i \in \agents$ has a valuation function $v_i : 2^\items \mapsto \mathbb{R}$, which maps each subset of items $S \subseteq \items$ to $v_i(S)$, the agent's value for $S$. 
We are interested in proving the existence of allocations that are fair with respect to the agent valuations. Our primary notion of fairness is envy-freeness.
 \begin{definition}[Envy-freeness]
     An allocation $\alloc = (A_1, \ldots, A_n)$ is \emph{envy-free (EF)} iff each agent prefers her own bundle over the bundle of any other agent, i.e., for all agents $i,j \in \agents$, we have $v_i(A_i) \geq v_i(A_j)$.
 \end{definition}

 For the case of indivisible items, envy-freeness can be too stringent of a requirement. A common relaxation is called envy-freeness up to any good, or EFX. 
 \begin{definition}[Envy-freenesss up to any good (EFX)]
     An allocation $\alloc$ is \emph{envy-free up to any good (EFX)} iff for all agents $i, j \in \agents$ where $A_j \ne \emptyset$, we have $v(A_i) \ge v(A_j \setminus \set{g})$ for all items $g \in A_j$.
 \end{definition}
 

In \Cref{sec:order}, we focus on a class of valuation functions that are consistent with an underlying preference order over the items. Given an ordering of the items $\pi = g_1 \succ g_2 \succ \cdots \succ g_m$, we can define a partial order over bundles $\sdpref_\pi$ as follows. Let $A,\hatA \subseteq \items$ be any two bundles such that $A = \{g_{\sigma_1}, g_{\sigma_2} \ldots, g_{\sigma_{|A|}} \}$ and $A = \{g_{\hatsig_1}, g_{\hatsig_2} \ldots, g_{\hatsig_{|\hatA|}} \}$ where sequences $\sigma$ and $\hatsig$ are sorted in an increasing order. We say that $A \sdpref_\pi \widehat{A}$ iff 
$|A| \geq |\hatA|$ and for all $k \in \{1, 2, \ldots, |\hatA|\}$, we have $\sigma_k \le \hatsig_k$. 
Additionally, $A \sdprefneq_\pi \widehat{A}$ iff $A \sdpref_\pi \widehat{A}$ and $A \neq \hatA$.

\begin{definition}[Order-consistency]
A valuation function $v: 2^\items \mapsto \mathbb{R}$ is called \emph{order-consistent with respect to an ordering $\pi$} iff for all bundles $A,B \subseteq \items$,  $A \sdpref_\pi B$ implies that $v(A) \geq v(B)$. If $v$ is order consistent with respect to some order $\pi$, we will simply call $v$ order-consistent.
\end{definition}




\begin{definition}[$\SDEF$]
Suppose that each agent $i \in \agents$ has a valuation function $v_i$ which is order-consistent with respect to $\pi_i$. An allocation $\alloc$ is \emph{sd-envy-free (SD-EF)}  iff for all $i,j \in \agents$ we have $A_i \sdpref_{\pi_i} A_j$.
\end{definition}

$\SDEF$ allocations are also $\EF$. In fact,   an allocation $\alloc$ which is $\SDEF$ for agents with valuation functions $v_1, v_2, \ldots, v_n$ that are order-consistent with respect to $\pi_1, \pi_2, \ldots, \pi_n$, respectively, is  $\EF$ for agents all  valuation functions $\widehat
{v}_1, \widehat{v}_2, \ldots, \widehat{v}_n$ that are order-consistent with respect to $\pi_1, \pi_2, \ldots, \pi_n$ respectively. 
When an allocation $\alloc = (A_1, A_2, \ldots, A_n)$ is $\SDEF$,  it follows that  $|A_i| = |A_j|$ for all $i,j \in \agents$, i.e., the number of items are a multiple of the number of agents, $m = qn$ for some integer $q \in \mathbb{N}$. When $m \neq qn$ (for $q\in \mathbb{N}$), we  consider the analogous strengthening of EFX.

\begin{definition}[$\SDEFX$]
    Suppose   each agent $i \in \agents$ has a valuation function $v_i$ that is order-consistent with respect to $\pi_i$. An allocation $\alloc$ is \emph{sd-envy-free up to any item} ($\SDEFX$) iff for all agents $i,j \in \agents$, where $A_j \neq \emptyset$, we have $A_i \sdpref_{\pi_i} A_j\setminus \{g\}$ for all items $g \in A_j$.
\end{definition}

As with SD-EF and EF,  an  allocation $\alloc$  which is SD-EFX is also $\alloc$ is EFX. While the existence of $\EFX$ remains unknown, in Appendix \ref{appendix:no-sdefx}, we give an instance for which $\SDEFX$ allocations do not exist.


Finally, if $v$ is a valuation function and $\pi$ is a permutation over items, we will use the notation $v^\pi$ to represent the permuted valuation function after random renaming, where $v^{\pi}(S) \coloneqq v(\pi^{-1}(S))$, and $\pi(S) \coloneqq \set{\pi(g): g \in S}$. If $v$ is an order-consistent valuation, then    $v^\pi$ is   order-consistent   with respect to  $\pi$.



\section{Arbitrary Valuations}\label{sec: general}

In this section, we consider the case of two agents having arbitrary set functions as their valuation functions. Such valuation functions need not be monotone and are general enough to capture all well-studied settings, including the fair division of goods, chores, or mixed manna under additive, subadditive, or superadditive preferences, and so on. Our main result is that $\EF$ allocations exist with probability at least $1-\frac{1}{m/2 + 1}$ when the number of items, $m$, is even. 

When the number of items is odd, there are valuation functions for which $\EF$ allocations cannot exist. For example, consider the case of two identical, additive agents whose value for every single item is equal to $1$, where random renaming has no effect. For such a valuation function, if $m$ is odd, then envy-free allocations do not exist. And, if $m$ is even, every envy-free allocation is, in fact, \emph{balanced}, i.e., each agent receives a bundle of the same size $m/2$. In~\Cref{theorem:ef-whp-gen}, we show that balanced EF allocations exist with high probability for arbitrary valuation functions when $m$ is even. As we'll see later in this section, looking for a balanced EF allocation will allow us to reduce our problem to a question about the size of automorphisms in $k$-uniform hypergraphs.
First, in the following observation, we show that, without loss of generality, we may assume that each agent has a strict preference between a set and its complement.





\begin{observation}\label{observation:balanced}
    If any agent has a valuation function $v:2^\items \to \mathbb{R}_{\geq 0}$ such that, for some $S \subseteq \items$, it holds that $v(S) = v(\overline{S})$, then  an $\EF$ allocation exists with probability $1$.
\end{observation}

\begin{proof}
    Without loss of generality, let $v$ be the valuation function of agent $1$. For any permutation $\pi_1$, agent $1$ is indifferent between the bundles $\pi_1(S)$ and $\pi_1(\bar{S})$, therefore, either the allocation $(\pi_1(S), \pi_1(\bar{S}))$ or $(\pi_1(\bar{S}), \pi_1(S))$ is envy-free. 
\end{proof}

We henceforth assume that $(i)$ the number of items, $m$, is even; and $(ii)$ the valuation function $v_i$ of each agent $i$ is such that, for every subset $S$ with $|S| = m/2$, $v_i(S) \neq v_i(\bar{S})$. Our proof is based on two key insights. First, (strict) preferences for bundles of size $m/2$ can be conveniently  represented as uniform hypergraphs. Second, we can formulate questions about the probability of an envy-free allocation existing as   questions about the number of hypergraph automorphisms. 

\subsection{Representing preferences as hypergraphs}
We use the following hypergraph representation of arbitrary set function preferences. Given a valuation function $v$ such that $v(S) \neq v(\bar{S})$ for all $S\subseteq \items, |S| = m/2,$ let $\hgraph^v = (\items, E_v)$ be the hypergraph where $E_v \coloneqq \{S \subset \items : |S| = m/2 ~~ \& ~~ v(S) > v(\overline{S})\}$. $\hgraph^v$ is a $m/2$-uniform hypergraph\footnote{A hypergraph $\hgraph = (V,E)$ is $k$-uniform when  $|e| = k$ for all $e \in E$.} with $m$ vertices, one for each item, and $\binom{m}{m/2}/2$ edges, one for each set of size $m/2$  that is preferred to its complement. 

Given a hypergraph $\hgraph = (\items, E)$ and a permutation $\pi: \items \to \items$, we will use $\pi(\hgraph) = (\items, E_\pi)$ to denote the hypergraph obtained by permuting the set of vertices by $\pi$, i.e., $E_\pi = \{\pi(S) : S \in E\}$. As an example, consider the valuations and corresponding hypergraphs shown in Figure~\ref{fig:example-hypergraph}.  Each preferred bundle of size 2 corresponds to an edge in the hypergraph, representing the corresponding agent's preferences. Since $\hgraph^1 \neq \hgraph^2$ there exists an envy-free allocation, in this case $(\{1,3\}, \{2,4\})$. 

\begin{figure}[htb]
    \centering
    \begin{tabular}{cc}
    \toprule
    $S$ & $\bar{S}$ \\
    \midrule
         ${\boxed{\underline{1,2}}}$ & $3,4$ \\
         $\boxed{1,3}$ & \underline{$2,4$} \\
         $\boxed{\underline{1,4}}$ & $2,3$ \\
    \bottomrule
    \end{tabular} $\quad$
    \adjustbox{valign=c}{
    \begin{tikzpicture}
        [
            he/.style={draw,  semithick},        
            ce/.style={draw, densely dotted, semithick}, 
            de/.style={draw, densely dashed, semithick}, 
        ]
        \node (5) at (-2,) {};
        \node (6) at (-1,1) {$H^{1}:$};
        \node (4) at (0,0) {$4$}; 
        \node (1) at (0,1) {$1$};
        \node (2) at (1,1) {$2$};
        \node (3) at (1,0) {$3$};
        
        \draw[ce] \hedgeii{1}{2}{3mm};
        \draw[ce] \hedgeii{1}{3}{3mm};
        \draw[ce] \hedgeii{1}{4}{3mm};
        \vspace{-1cm}
    \end{tikzpicture}
    \begin{tikzpicture}
        [
            he/.style={draw,  semithick},        
            ce/.style={draw, densely dotted, semithick}, 
            de/.style={draw, densely dashed, semithick}, 
        ]
        \node (5) at (-2,0) {};
        \node (6) at (-1,1) {$H^{2}:$};
        \node (4) at (0,0) {$4$}; 
        \node (1) at (0,1) {$1$};
        \node (2) at (1,1) {$2$};
        \node (3) at (1,0) {$3$};
        
        \draw[ce] \hedgeii{1}{2}{3mm};
        \draw[ce] \hedgeii{2}{4}{3mm};
        \draw[ce] \hedgeii{1}{4}{3mm};
        \vspace{-1cm}
        \end{tikzpicture}
        \begin{tikzpicture}
        [
            he/.style={draw,  semithick},        
            ce/.style={draw, densely dotted, semithick}, 
            de/.style={draw, densely dashed, semithick}, 
        ]
        \node (5) at (-2,0) {};
        \node (6) at (-1,1) {$\pi(H^{1}):$};
        \node (4) at (0,0) {$4$}; 
        \node (1) at (0,1) {$1$};
        \node (2) at (1,1) {$2$};
        \node (3) at (1,0) {$3$};
        
        \draw[ce] \hedgeii{1}{2}{3mm};
        \draw[ce] \hedgeii{2}{4}{3mm};
        \draw[ce] \hedgeii{2}{3}{3mm};
        \vspace{-1cm}
        \end{tikzpicture}
    }
    \caption{Example with $m=4$ items and $n=2$ agents where the hypergraphs corresponding to the valuations of agent $1$, $2$ are $\hgraph^{1}$, $\hgraph^{2}$ respectively.
    The leftmost table shows all bundles of size $m/2 =2$ and their complements; the preferred bundles of   agent $1$ ($2$)   are boxed (underlined, respectively).
    The rightmost figure shows the hypergraph $\pi(\hgraph^1)$ obtained by renaming $\hgraph^1$ using permutation $\pi = (2,1,3,4)$.
    } 
    \label{fig:example-hypergraph}
\end{figure}
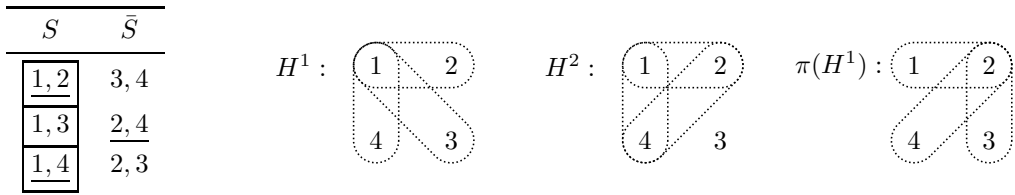


The next lemma formalizes the relationship between envy-free allocations and the hypergraph representation of valuations.


\begin{lemma}\label{lemma:hgraph-and-ef}
    Let $v_1, v_2 : 2^\items \to \mathbb{R}_{\geq 0}$ be two valuation functions with corresponding hypergraphs $\hgraph^{v_1} = (\items, E_{v_1})$ and $\hgraph^{v_2} = (\items, E_{v_2})$, respectively. If the hypergraphs are not identical, i.e. $\hgraph^{v_1} \neq \hgraph^{v_2}$, then an envy-free allocation exists.
\end{lemma}
\begin{proof}
    Since $\hgraph^{v_1} \neq \hgraph^{v_2}$, there must exist an edge $S \subset \items$, $|S| = m/2$, present in one of the hypergraphs that is not present in both. 
      If $S \in E_{v_1}$ and $S \notin E_{v_2}$, then  the allocation $\alloc = (S, \overline{S})$ is envy-free, since $v_1(S) > v_1(\overline{S})$ and $v_2(S) < v_2(\overline{S})$. Similarly, if $S \notin E_{v_1}$ and $S \in E_{v_2}$, then the allocation $\mathcal{B} = (\overline{S}, S)$ is envy-free. 
      Thus, envy-free allocations exist when $\hgraph^{v_1} \neq \hgraph^{v_2}$.
\end{proof}

\subsection{Relating envy-freeness to hypergraph isomorphism}

When agent $i$ with valuation function $v_i$ and corresponding hypergraph $\hgraph^{v_i}$ renames the items by applying permutation $\pi$, then the hypergraph $\hgraph^{v_i^\pi}$ corresponding to the new valuation function $v^\pi$ is   $\hgraph^{v_i^\pi} = \pi(\hgraph^{v_i})$. That is, permuting the names of the items amounts to permuting the vertices of the hypergraph. This follows directly from the following sequence of equivalences: subset $S$ is present as an edge in $\hgraph^{v_i^\pi}$ iff $v_i^\pi(S) > v_i^\pi(\overline{S})$ iff $v_i(\pi^{-1}(S)) > v_i(\pi^{-1}(\overline{S}))$ (from the definition of $v_i^\pi$) iff $\pi^{-1}(S)$ is present as an edge in $\hgraph^{v_i}$ iff $S$ is present as an edge in $\pi(\hgraph^{v_i})$.

In our setting, agents $1$ and $2$, with valuation functions $v_1$ and $v_2$, apply permutations $\pi_1$ and $\pi_2$ to rename the items. As per \Cref{lemma:hgraph-and-ef}, an envy-free allocation will exist whenever $\pi_1(\hgraph^{v_1}) \neq \pi_2(\hgraph^{v_2})$. Since applying permutations $\pi_1, \pi_2$ does not change the structure of the hypergraph, 
$\pi_1(\hgraph^{v_1}) = \pi_2(\hgraph^{v_2})$ is only possible when 
$ \hgraph^{v_1}$ and $\hgraph^{v_2}$ are isomorphic (and therefore, envy-free allocations trivially exist with probability 1 whenever $\hgraph^{v_1}$ and $\hgraph^{v_2}$ are non-isomorphic). 

Accordingly, we  assume that $\hgraph^{v_1}$ and $\hgraph^{v_2}$ are isomorphic. 
%
Without loss of generality, we also assume that $\pi_1$ is the identity permutation. To see why, observe that for uniformly random $\pi_1, \pi_2$, the composition $\pi_1^{-1}\pi_2$ is also uniformly random. 
We view agent 2's renaming of the items as the permutation group\footnote{Recall that the elements of the permutation group are all permutations $\pi : \items \to \items$ and the group operation which combines two group elements is the function composition.} 
acting on the set of all hypergraphs on $m$ vertices,  denoted $\hgraphset_m = \{ H = (V,E): |V| = m \}$ and recall relevant group theoretic concepts. 
The \emph{orbit} of a hypergraph $\hgraph$ is an equivalence class defined as $\mathcal{O}(H)  = \{ G\in \hgraphset_m: \exists \text{ permutation }\pi, G =  \pi(\hgraph)\}.$
The permutation group partitions $\hgraphset_m$ into  $\ell$ orbits $\mathcal{O}_1, \mathcal{O}_2, \ldots, \mathcal{O}_\ell$, where any two hypergraphs are isomorphic if and only if  they are in the same orbit. 
In Figure \ref{fig:example-hypergraph}, $\hgraph^1$ and $\hgraph^2$ are not in the same orbit. Despite $\hgraph^1$ and $\pi(\hgraph^1)$ being isomorphic and in the same orbit, these valuations still permit an envy-free allocation. 



Let $\mathcal{O}$ be the orbit containing both $\hgraph^{v_1}$ and $\hgraph^{v_2}$. 
Starting from  $\hgraph \in \mathcal{O}$ and applying a random permutation $\pi$,   the distribution of the permuted hypergraph $\pi(H)$ is also uniform over $\mathcal{O}$.
To upper bound the probability of the non-existence of $\EF$ allocations,  it is sufficient to lower bound all  orbit sizes.
Specifically, $\Pr_{\pi_2}[ \hgraph^{v_1} = \pi_2(\hgraph^{v_2})] = 1/|\mathcal{O}|$. 
%
The orbit-stabilizer theorem~\cite{dummit2004abstract} relates a graph's orbit size to the number of  its stabilizers.

\begin{theorem}[Orbit-stabilizer theorem \cite{dummit2004abstract}] \label{thm:ost}
    If a finite group $\Pi$ acts on a set $X$, then for every $x\in X$,
    $|\Pi| = |\orb(x)| \cdot |\stab(x)|,$ where $\stab(x) = \{\pi\in \Pi: \pi\cdot x = x\}.$
\end{theorem}
In our context $\Pi$ is the permutation group and  $X = \hgraphset_m$. 
 A  permutation $\pi$ is an automorphism (and also a stabilizer) of  hypergraph $\hgraph$ iff $\pi(\hgraph) = \hgraph$.
Let $\aut(\hgraph) \coloneqq \{ \text{ permutation } \pi : \pi(\hgraph) = \hgraph \}$ be the set of all automorphisms. Now  $m! = |\mathcal{O}(H)| \cdot |\aut(H)|$ by Theorem~\ref{thm:ost}.
%
%
We now upper bound the number of automorphisms of a hypergraph $H$.



\begin{theorem}\label{theorem:auto-bound}
    For any non-empty and non-complete $k$-uniform hypergraph $\hgraph$ on $m$ vertices, the number of automorphisms of $\hgraph$ satisfies $|\aut(\hgraph)| \leq \frac{m!}{m - k + 1}$.
\end{theorem}
\begin{proof}
We prove this by induction on $k$. Our induction hypothesis is as follows: for a given $k$, the number of automorphisms of a non-empty and non-complete $k$-uniform hypergraph on $m$ vertices is at most $\frac{m!}{m - k + 1}$.

    \noindent
    \textbf{Base Case: $k = 1$}. Fix an arbitrary $m$. A $1$-uniform hypergraph $\hgraph = (\items, E)$ is simply a selection of some $\ell$ singleton sets of vertices: for all $S \in E$, we have $|S| = 1$. Since $\hgraph$ is non-empty and non-complete, we have $0 < \ell < m$. Note that, any automorphism $\pi$ of $\hgraph$ must map any vertex that forms (resp. does not form) an edge to another vertex that forms (resp. does not form) an edge. Since there are $\ell$ vertices that form an edge and $m - \ell$ vertices that don't form an edge, we have
    \[|\aut(\hgraph)| \leq \ell!(m - \ell)! = \frac{m!}{\binom{m}{\ell}} \le \frac{m!}{m} = \frac{m!}{m - k + 1}.
    \]
    
    \noindent
    \textbf{Induction Step:} Now, assume that the induction hypothesis is true for $k$. Fix an arbitrary $m$ and let $\hgraph = (\items, E)$ be a non-empty, non-complete $(k+1)$-uniform hypergraph on $m$ vertices. For any vertex $i \in \items$, define $\hgraph_i$ to be the $k$-uniform hypergraph induced by the vertex $i$. That is, $\hgraph_i = (V_i, E_i)$ where $V_i = \items \setminus \set{i}$ and $E_i = \set{e \setminus \set{i} : e \in E,  i \in e}$. We consider two cases, based on whether the hypergraphs $\set{\hgraph_i}_{i \in \items}$ are all isomorphic to each other or not.
    
    First, suppose the hypergraphs $\set{\hgraph_i}_{i \in \items}$ are not isomorphic to each other. Define $I_1 = \{j \in \items: \exists \text{ permutation } \pi \text{ satisfying }\pi(\hgraph_j) = \hgraph_1 \}$  as the set of vertices whose corresponding hypergraph is isomorphic to $\hgraph_1$ and let $\ell = |I_1|$. $1 \in I_1$, therefore $\ell > 0$. Since not all $H_i$'s are isomorphic to each other we have $I_1 \subsetneq \items$, and thus $\ell < m$. In any automorphism $\pi$, every vertex in $I_1$ must be mapped to a vertex in $I_1$. Hence, by the same argument as in the base case, the total number of automorphisms of $\hgraph$ is at most
    \[
        |\aut(\hgraph)| \leq \ell!(m - \ell)! \le \frac{m!}{m} \le \frac{m!}{m - k + 1}.
    \]
    
    Next, suppose the hypergraphs $\set{\hgraph_i}_{i \in \items}$ are isomorphic to each other. Because $\hgraph$ is non-empty and non-complete and $k \ge 2$, each $\hgraph_i$ cannot be empty or complete. That is, the hypergraphs $\set{\hgraph_i}_{i \in \items}$ are non-empty, non-complete, $k$-uniform, and have $m - 1$ vertices. Therefore, as per the induction hypothesis, we have $|\aut(\hgraph_j)| \leq \frac{(m - 1)!}{(m - 1) - k + 1}$ for all $j \in \items$. Finally, to count the total number of automorphisms of $\hgraph$, notice that there are exactly $m$ ways to map vertex $1$ to another vertex $j \in \items$. However, once this mapping is fixed, the hypergraph induced by the remaining vertices $\items \setminus \set{j}$, must get mapped to the hypergraph $\hgraph_1$. Hence, the total number of automorphisms of $\hgraph$ is at most
    \[
        |\aut(\hgraph)| \leq  m \cdot \frac{(m - 1)!}{(m - 1) - k + 1} = \frac{m!}{m - (k + 1) + 1},
    \]
    giving us the desired inequality. This completes the induction step.
\end{proof}

We combine the previous observations and~\Cref{theorem:auto-bound} to obtain the main result of this section.

\begin{theorem}\label{theorem:ef-whp-gen}
For $n=2$ agents, having arbitrary valuation functions, envy-free allocations exist with a probability of at least $1-\frac{1}{m/2+1}$ when $m$ is even. 
This probability is tight up to constants: for $n=2$ agents, there exist additive valuations such that an envy-free allocation does not exist with probability $1/m$.
\end{theorem}

\begin{proof}
From Observation~\ref{observation:balanced} and Lemma~\ref{lemma:hgraph-and-ef}, we only need to consider instances where agents have valuation functions $v_1$ and $v_2$ such that $(i)$ $\hgraph^{v_1}$ is isomorphic to $\hgraph^{v_2}$, and $(ii)$ for all bundles $S \subseteq \items$ with $|S| = m/2$ we have $v(S) \neq v(\overline{S})$; if either $(i)$ or $(ii)$ is not satisfied, $\EF$ allocations will exist with probability $1$. We can therefore represent the valuation functions by the corresponding hypergraphs $\hgraph^{v_1}$ and $\hgraph^{v_2}$. Since  $\hgraph^{v_1}$ is isomorphic to $\hgraph^{v_2}$, instead of considering renaming $\hgraph^{v_1}$ and $\hgraph^{v_2}$ by random permutations $\hat{\pi}_1$ and $\hat{\pi}_2$ respectively, we can simply consider $v_1$ as fixed and permuting $\hgraph^{v_2}$ by $\pi = \hat{\pi}^{-1}_1 \hat{\pi}_2$.

By \Cref{lemma:hgraph-and-ef}, envy-free allocations do not exist with probability at most $\Pr_{\pi}[\hgraph^{v_1} = \pi(\hgraph^{v_2})]$. If both hypergraphs belong to the orbit $\mathcal{O}$,  i.e., $\hgraph^{v_1}, \hgraph^{v_2} \in \mathcal{O}$, then $\Pr_{\pi}[\hgraph^{v_1} = \pi(\hgraph^{v_2})] = \frac{1}{|\mathcal{O}|} = \frac{|\aut(\hgraph^{v_1})|}{m!}$; the last inequality follows from the orbit-stabilizer theorem~\cite{dummit2004abstract}. As the final step, we can use \Cref{theorem:auto-bound} with $k=m/2$ to obtain the bound
\[
\Pr_{\pi}[\hgraph^{v_1} = \pi(\hgraph^{v_2})] = \frac{|\aut(\hgraph^{v_1})|}{m!} \leq \frac{1}{m - (m/2) + 1} = \frac{1}{m/2 + 1}.
\]
Therefore, envy-free allocations exist with a probability of at least $1 - \frac{1}{m/2 + 1}$. 

To conclude the proof of the theorem, we show an upper bound on the probability of existence. Consider the case of two identical valuation functions that assign a positive value only to a single item (the same item for both agents). Any set that contains this item is preferred to its complement. For such an instance, with probability $1/m$ (over item renaming), this single item liked by both of the agents will be the same, and envy-free allocations won't exist.
\end{proof}

\section{Order-Consistent Valuations}\label{sec:order}
Our second main result is that $\SDEF$ allocations exist with high probability for order-consistent valuation functions. In fact, the well-known Round-Robin algorithm, presented as Algorithm~\ref{alg:rr}, outputs such allocations (with high probability) in polynomial time. Round-Robin takes as input an ordering of the items for each agent, and agents take turns taking the best available remaining item according to this order. As long as Round-Robin has access to these orderings after renaming (that is, no other information about the valuation functions is needed), and as long as $m$ is divisible by $n$, Round-Robin will result in an $\SDEF$ allocation with high probability. The precise probability is asymptotically optimal. Beyond this, if $m$ is not divisible by $n$ (so there is no hope for an $\EF$ allocation), Round-Robin will find an $\SDEFX$ allocation with, once again, asymptotically optimal probability.

\begin{algorithm}
\caption{Round Robin Algorithm} \label{alg:rr}
\DontPrintSemicolon
\Input{Valuations $v_1, \ldots, v_n$ which are order consistent with respect to $\pi_1, \ldots, \pi_n$.}
\Output{ Allocation $\alloc = (A_1, A_2, \ldots, A_n)$}
\hrulefill
\\
\textbf{Set} $A_i \gets \emptyset$ for all agents $i \in \agents$\;
\textbf{Set} $P \gets \items$\;
\For{$i = 1, \ldots, n, 1, \ldots, n, 1, \ldots$}{
    {\bf Let} $g \in P$ be the unallocated item with the lowest index in $\pi_i$.\;
    {\bf Set} $A_i \gets A_i \cup \{g\}$\;
    {\bf Set} $P \gets P \setminus \{g\}$\;
}
\Return{$(A_1, A_2, \ldots, A_n)$}
\end{algorithm}

\begin{theorem}\label{thm:positive for order consistent}
    Let $\pi_1, \ldots, \pi_n$ be sampled independently and uniformly at random. When  $m = qn$, $q\in \mathbb{N}$,
    \[
        \Prs{\text{Round Robin is $\SDEF$}} \ge 1 - O\left(\frac{n^2}{m} + \frac{n \log m}{m^{\frac{n - 1}{n}}} \right),
    \]
        and, if $m \ge 2n$ (where $m$ may or may not be a multiple of $n$)
    \[
        \Prs{\text{No $\EF$ allocation exists} }\in \Omega\left(\min\left(\frac{n^2}{m}, 1\right)\right).
    \]
    Further, for any $m$,
    \[
        \Prs{\text{Round Robin is $\SDEFX$}} \ge 1 - O\left(\frac{n^2}{m}\right).
    \]

\end{theorem}
Notice that  if $m \in \omega(n^2)$, then, as $n$ grows, $\SDEF$/$\SDEFX$ allocations exist with high probability, while the lower bound implies that even $\EF$ allocations do not exist when $m \in O(n^2)$. In this sense, Round-Robin is asymptotically optimal. 


For ease of exposition, we break the proof of~\Cref{thm:positive for order consistent} into separate lemmas.  We first assume that  $m = qn$ for some integer $q$, and prove the existence of $\SDEF$ allocations. 

\begin{lemma}[Round-Robin is $\SDEF$]\label{lem: round robin is ef}
When  $m = qn$, 
    $\Prs{\text{Round Robin is $\SDEF$}} \ge 1 - O\left(\frac{n^2}{m} + \frac{n \log m}{m^{\frac{n - 1}{n}}} \right)$.
\end{lemma}

Second, we show how to extend our arguments to the case of arbitrary $m$  and show the existence of $\SDEFX$ allocations. 

\begin{lemma}[Round-Robin is $\SDEFX$]\label{lem: rr is sdefx}
For any $m$, $\Prs{\text{Round Robin is $\SDEFX$}} \ge 1 - O\left(\frac{n^2}{m}\right)$.
\end{lemma}

Finally, we show the lower bound on the probability that $\EF$ or $\SDEFX$  allocations exist. 

\begin{lemma}
\label{lem: rr lower bound}
For any $m$, if $m \ge 2n$, $\Prs{\text{No $\EF$ allocation exists} }\in \Omega\left(\min\left(\frac{n^2}{m}, 1\right)\right)$. The same bound holds for the probability that no $\SDEFX$ exists.
\end{lemma}

The three lemmas combined imply Theorem~\ref{thm:positive for order consistent}.
We prove Lemma~\ref{lem: round robin is ef} in~\Cref{subsec: RR EF}. Lemma~\ref{lem: rr is sdefx} largely follows the same proof; we prove it in~\Cref{subsec: RR is EFX}. Finally, Lemma~\ref{lem: rr lower bound} is proved in~\Cref{subsec: RR is asymptotically optimal}.

In~\Cref{subsec: give away}, we observe that, even though Round-Robin is asymptotically optimal as $n$ grows, for the important case of small (i.e., constant) $n$, it leaves a small gap in our understanding. We introduce a new algorithm, Give-Away Round Robin, which closes this gap.


\subsection{Round-Robin is $\SDEF$: The proof of Lemma~\ref{lem: round robin is ef}}\label{subsec: RR EF}

Our goal will be to show that the probability that Round-Robin does not output an $\SDEF$ allocation is at most $\frac{33n^2}{m} + \frac{16n\log m}{m^{1 - 1/n}}$. We restrict our analysis to $m$ and $n$ such that $m \ge 2n$ and $\frac{n\log m}{m^{n - 1/n}} \le 1/16$, since otherwise, our target upper bound holds trivially. Fix $m$, $n$, and $q$ such that $m = qn$.  
    In each round of the Round-Robin algorithm, lower-indexed agents pick before higher-indexed ones, and therefore the former do not envy the latter. That is, agent $i$ can only envy agent $j$'s bundle if  $j < i$.
    For agents $i, j \in \agents$ with $j < i$, let $\cE^{ij}$ be the event that $A_i \notsdpref_{\pi_i} A_j$, i.e., agent $i$ does not sd-prefer their own bundle over that of agent $j$. The probability that the final allocation is not SD-EF is then exactly the probability that one of these events occurs; formally, $\Prs{\bigcup_{i \in \agents} \bigcup_{j: j < i} \cE^{ij}}$. 
    
    
    We use $g \succ_i g'$ to denote that item $g$ is preferred to item $g'$ under $\pi_i$. For a set of items $S \subseteq \items$, we   write $g \succ_i S$ when $g$ is preferred to all items in $S$, i.e., $g \succ_i g'$ for all $g' \in S$. Additionally, we use $P_t$ to denote the pool of available items at ``time'' $t$, i.e., after $t$ picks have been made (so, $P_0 = \items$ and $P_m = \emptyset$). Let $g_t$ be the $t$'th item picked during the execution of the Round Robin algorithm, that is, $P_{t + 1} = P_t \setminus \set{g_{t + 1}}$. Since agents pick items sequentially as per their index, $g_t$ is picked by agent $j$ exactly when $t \equiv j \bmod n$, $g_{(\ell - 1)n + j}$ is the $\ell$'th item picked by agent $j$, and $A_j = \set{g_{0 \cdot n + j}, \ldots, g_{(q - 1)n + j}}$.
    
  Taking the perspective of agent $i$, let $b^i_{jk}$ denote the $k$'th best item according to $\pi_i$ in the final bundle $A_j$; we will write $b_{jk}$ if the agent $i$ is clear from the context. For agent $i$'s picks, $b_{ik} = g_{(k - 1)n + i}$, i.e., their $k$'th pick is also their $k$'th favorite item in their bundle. However, this is not the case for items in others' bundles. Since the final bundles are all of size $q$, $A_i \sdpref_{\pi_i} A_j$ exactly when $b_{ik} \succ_i b_{jk}$ for all $k \le q$. 
  Hence, we decompose the event $\cE^{ij}$ into events $\cE^{ij}_k$, where $\cE^{ij}_k$ is the event that $b_{jk} \succ_i b_{ik}$. Now $\cE^{ij} = \bigcup_{k = 1}^q \cE^{ij}_k$, and (by a union bound) our goal will be to upper bound
    \begin{equation}\label{eq: sum of probabilities}
        \sum_{i \in \agents} \sum_{k = 1}^q \Prs{\bigcup_{j: j < i} \cE^{ij}_k}.
    \end{equation}

    Fix an agent $i$ and value $k$. To get a handle on $\Prs{\bigcup_{j \in \agents: j < i} \cE^{ij}_k}$, we condition on agent $i$'s ordering of the items, $\pi_i$. Whenever agent $i$ picks $g_t$, $g_t$ must be the best item remaining in $P_{t - 1}$ according to $\pi_i$.
    When agent $j \ne i$ picks $g_t$ (so $t \not\equiv i \pmod{n}$), the distribution of $g_t$ over the random draws of $\pi_{-i} = \{\pi_1, \ldots, \pi_n\} \setminus \{\pi_i\}$, even conditioned on all previous picks $g_1, \ldots, g_{t - 1}$, is uniformly random over $P_{t - 1}$.
    Let $R_\ell$ be the set of items picked between agent $i$'s $\ell$'th and $(\ell + 1)$'th pick. That is, $R_0 = \set{g_1, \ldots, g_{i - 1}}$, for $1 \le \ell \le q - 1$, $R_\ell = \set{g_{n (\ell - 1) + i + 1}, \ldots, g_{\ell n + i - 1}}$, and $R_q = \set{g_{n(q - 1) + i+1}, \ldots, g_m}$.  From the perspective of agent $i$, each item in  $R_\ell$ is picked sequentially and uniformly at random from the remaining items. This allows us to analyze an equivalent two-step process wherein (1) $R_\ell$ is first selected uniformly at random from $P_t$, after which, (2) each item in $R_\ell$ is matched to a uniformly random agent (different than $i$) that picked in the corresponding timesteps. 

    We upper bound the expression in~\eqref{eq: sum of probabilities} by upper bounding each of the summands. Our analysis splits into cases based on the value of $k$, as care must be taken when comparing $i$'s worst item to that of other agents.
    
    \begin{lemma}\label{lem: bound for k < q, EF}
    For $k < q$, it holds that $\Prs{\bigcup_{j: j < i} \cE^{ij}_k} \le \frac{8}{\binom{q}{k}}$.
    \end{lemma}

    \begin{lemma}\label{lem: bound for k = q, EF}
    For $k = q$, it holds that $\Prs{\bigcup_{j: j < i} \cE^{ij}_q} \le
     \frac{n}{m} + 2\left(\frac{8n\log m}{m^{1 - 1/n}}\right)^{n + 1 - i}$.
    \end{lemma}

    We show how to conclude the proof of~\Cref{lem: round robin is ef} given these two lemmas,   then proceed to prove them. Using Lemmas~\ref{lem: bound for k < q, EF} and~\ref{lem: bound for k = q, EF}, we have
    \begin{align*}
        \Prs{\bigcup_{i,j} \cE^{ij} } &\le \sum_{i \in \agents} \left(\sum_{k = 1}^{q - 1}\Prs{\bigcup_{j} \cE^{ij}_k } + \Prs{\bigcup_{j} \cE^{ij}_q }  \right)\\
        &\le \sum_{i \in \agents} \left( \sum_{k = 1}^{q - 1}\frac{8}{\binom{q}{k}} + \frac{n}{m} + 2\left(\frac{8n\log m}{m^{1 - 1/n}}\right)^{n + 1 - i} \right)\\
        &=8n \sum_{k = 1}^{q - 1}\frac{1}{\binom{q}{k}} + \frac{n^2}{m} +  2 \sum_{i = 1}^n \left(\frac{8n\log m}{m^{1 - 1/n}}\right)^{n + 1 - i}.
    \end{align*}
    We now bound each of these sums. For the first,
    \begin{equation}\label{eq:q-sum}
        \sum_{k = 1}^{q - 1}\frac{1}{\binom{q}{k}} \le \frac{1}{\binom{q}{1}} + \frac{1}{\binom{q}{q - 1}} + \sum_{k = 2}^{q - 2} \frac{1}{\binom{q}{k}} \le \frac{2}{q} + \sum_{k = 2}^{q - 2} \frac{1}{\binom{q}{2}} = \frac{2}{q} + \frac{2(q - 3)}{q(q - 1)} \le \frac{4}{q} = \frac{4n}{m}.
    \end{equation}
    For the second, since $\frac{8n\log m}{m^{1 - 1/n}} \le 8/16 = 1/2$,
    \begin{align*}
        \sum_{i = 1}^n \left(\frac{8n\log m}{m^{1 - 1/n}}\right)^{n + 1 - i}
        &= \sum_{i = 1}^n \left(\frac{8n\log m}{m^{1 - 1/n}}\right)^{i} \le \sum_{i = 1}^\infty \left(\frac{8n\log m}{m^{1 - 1/n}}\right)^{i} = \frac{\frac{8n\log m}{m^{1 - 1/n}}}{1 - \frac{8n\log m}{m^{1 - 1/n}}} \le \frac{16n\log m}{m^{1 - 1/n}}.
    \end{align*}
    We conclude that
    \[
    \Prs{\bigcup_{i,j} \cE^{ij} } \le
        \frac{33n^2}{m} + \frac{16n\log m}{m^{1 - 1/n}} \in O\left(\frac{n^2}{m} + \frac{n \log m}{m^{\frac{n - 1}{n}}} \right),
    \]
    as needed.

    \begin{proof}[Proof of Lemma~\ref{lem: bound for k < q, EF}]
    We will  directly upper bound each $\Prs{\cE^{ij}_k}$ and union bound over the (at most) $n$ possible choices of $j$. 
    The equivalent random process of picking $R^{\ell}$ and   assigning each item randomly discussed above implies $\Prs{\cE^{ij}_k} = \Prs{\cE^{ij'}_k}$, for each $j, j' < i$. Indeed, the items assigned to  $j$ and $j'$ in this process  appear in exactly the same $R_\ell$ sets. Now  the distributions of $A_j$ and $A_{j'}$ are identical even when conditioning on $A_i$. As a result, we  only upper bound $\Prs{\cE^{ij}_k}$ for $j = i - 1$; the same bound  holds for all agents that $i$ could possibly envy (i.e., for all agents $j'$ with $j' < i$). 
    
    Consider the item $g_{(k - 1)n + i}$, which is  picked by $i$ in the $k$'th round. As noted above, $g_{(k - 1)n + i} = b_{ik}$, so $\cE^{ij}_k$ occurs exactly when $b_{jk} \succ_i b_{ik} = g_{(k - 1)n + i}$. 
    This  can only occur if there are $k$ timesteps, $t_1,\dots, t_k$, such that $j$ picks an item $g_{t_z}$ satisfying $ g_{t_z} \succ_i g_{(k - 1)n + i}$. For any $t > (k - 1)n + i$, $g_{(k - 1)n + i} \succ_i g_t$,  since $i$ always picks the best available remaining item. Hence, all $k$ of these (unfortunate) picks must have occurred before $i$'s  $k$'th pick. Agent $j$ made exactly $k$ picks before $i$ made their $k$'th pick, so for $b_{jk} \succ_i g_{(k - 1)n + i}$ to hold, a necessary and sufficient condition is that all these picks  were  preferred to $g_{(k - 1)n + i}$ (according to $\pi_i$), i.e.,  $g_{(\ell - 1)n + j} \succ_i g_{(k - 1)n + i}$ for all $\ell = 1, \ldots, k$. 

    Let $P_t^{(r)}$ be the $r$'th best item remaining in $P_t$ according to $\pi_i$. For all $t < t'$, $|P_t \setminus P_{t'}| = t' - t$, i.e., there are $t'-t$ fewer items available after $t'-t$ picks. This implies that for all $t < (k - 1)n + i$, $g_{(k - 1)n + i} \succeq_i P_t^{((k - 1)n + i - t)}$.\footnote{And this is tight only when all picks $g_t, \ldots, g_{(k - 1)n + i-1}$  are preferred to $g_{(k - 1)n + i}$ in $\pi_i$.}
   Furthermore,  $g_{(k - 1)n + i} = P_{(k - 1)n + i - 1}^{(1)}$ as $i$ picks the best available item remaining at time $(k - 1)n + i - 1$.  
   $|P_t \setminus P_{((k - 1)n + i - 1)}| = (k - 1)n + i - t - 1$, so at least one of $P_t^{(1)}, \ldots, P_t^{((k - 1)n + i - t)}$ must be available at time $(k - 1)n + i - 1$, and the best available item at time $(k - 1)n + i - 1$ must be at least as good as the worst of these. This implies that if $g_{(\ell - 1)n + j} \succ_i g_{(k - 1)n + i}$, then $g_{(\ell - 1)n + j} \succ_i P^{((k - \ell)n + (i - j)+1)}_{(\ell - 1)n + j - 1}$. Using the fact that $j = i - 1$, this simplifies to $g_{(\ell - 1)n + j} \succeq_i P^{((k - \ell)n + 1)}_{(\ell - 1)n + j - 1}$
   or, in words, agent $j$'s $\ell$'th pick must be one of the top $(k - \ell)n + 1$ available items according to $\pi_i$. This necessary condition means that
    \[
        \Prs{\cE^{ij}_k} \le \Prs{\bigcap_{\ell = 1}^k \left\{g_{(\ell - 1)n + j} \succeq_i P^{((k - \ell)n + 1)}_{(\ell - 1)n + j - 1}\right\}}.
    \]
    The key observation is that each of the events $g_{(\ell - 1)n + j} \succeq_i P^{((k - \ell)n + 1)}_{(\ell - 1)n + j - 1}$, over the samples of $\pi_{-i}$, are mutually independent. Hence, we get
    \[
    \Prs{\bigcap_{\ell = 1}^k \left\{g_{(\ell - 1)n + j} \succeq_i P^{((k - \ell)n + 1)}_{(\ell - 1)n + j - 1}\right\}}
    = \prod_{\ell = 1}^k \Prs{g_{(\ell - 1)n + j} \succeq_i P^{((k - \ell)n + 1)}_{(\ell - 1)n + j - 1}}.
    \]
    Furthermore, $g_{(\ell - 1)n + j}$ is a uniform random sample from $P_{(\ell - 1)n + j - 1}$, so \[\Prs{g_{(\ell - 1)n + j} \succeq_i P^{((k - \ell)n + 1)}_{(\ell - 1)n + j - 1}} = \frac{(k - \ell)n + 1}{|P_{(\ell - 1)n + j - 1}|}.\]
    Noticing that $|P_t| = m - t$, we get that,
    \begin{align}
    	\Prs{\cE^{ij}_k}
    	&= \prod_{\ell=1}^k \frac{(k - \ell)n + 1}{|P_{(\ell - 1)n + j - 1}|} = \prod_{\ell=1}^k \frac{(k - \ell)n + 1}{m - ((\ell - 1)n + j - 1)} = \prod_{\ell=1}^k \frac{(k - \ell)n + 1}{(q - \ell)n + n - j + 1} \label{eq:m-to-q}\\
    	&\le \prod_{\ell=1}^k \frac{(k - \ell)n + 1}{(q - \ell)n + 1} = \prod_{\ell = 1}^k \frac{k - \ell + 1/n}{q - \ell + 1/n} = \frac{\prod_{\ell = 0}^{k - 1} (1/n + \ell)}{\prod_{\ell = 0}^{k - 1} (q - k + 1/n + \ell)} = \frac{(1/n)^{\overline{k}}}{(q - k + 1/n)^{\overline{k}}}\nonumber.
    \end{align}
    The inequality follows because $(n - j) \ge 0$, and the second to last equality follows by a change of variables. The notation $(x)^{\overline{r}}$ here represents the rising factorial $(x)^{\overline{r}} = x(x + 1) \cdots (x + r - 1)$. We make use of the known equality $(x)^{\overline{r}} = \frac{\Gamma(x + r)}{\Gamma(x)}$ to get that
    \[
        \Prs{\cE^{ij}_k} = \frac{\Gamma(k + 1/n)}{\Gamma(1/n)}\cdot \frac{\Gamma(q - k + 1/n)}{\Gamma(q + 1/n)} \leq \frac{2}{n} \cdot \frac{\Gamma(k + 1/n)\Gamma(q - k + 1/n)}{\Gamma(q + 1/n)},
    \]
    where we replaced $\Gamma(1/n)$ by using the fact that $\Gamma(1 + x) = x\Gamma(x)$ for all $x$ (which implies that $\Gamma(1/n) = n\Gamma(1 + 1/n)$), and that the $\Gamma$ function has a   global minimum on the positive reals above $1/2$~\cite{wrench1968concerning} (and thus, $n\Gamma(1 + 1/n) \ge \frac{n}{2}$).

    Ideally, we could replace the $1/n$ terms inside the gamma function with $1$s, which would allow us to convert the Gammas to the corresponding factorial terms, and would simplify to $
    1/\binom{q}{k}$. We show that this is approximately correct by making use of Gautchi's inequality~\cite{gautschi1959some}, which states
    \[
        x^{1 - s} \le \frac{\Gamma(x + 1)}{\Gamma(x + s)} \le (x + 1)^{1 - s},
    \]
    for all $x > 0$ and $s \in (0, 1)$. We state our bound as the following claim, the proof of which is deferred to Appendix~\ref{app: missing proofs}.

    \begin{claim}\label{claim: gamma stuff 1}
$\frac{2}{n} \cdot \frac{\Gamma(k + 1/n)\Gamma(q - k + 1/n)}{\Gamma(q + 1/n)} \leq \frac{8}{n} \cdot \frac{1}{\binom{q}{k}}$.
    \end{claim}
    
    Therefore, for all $j < i$ and $k < q$, $\Prs{\cE^{ij}_k} \le \frac{8}{n} \cdot \frac{1}{\binom{q}{k}}$. Union bounding over the (at most) $n$ possible values of $j$, we have that 
    \begin{equation}
    \Prs{\bigcup_{j: j < i} \cE^{ij}_k} \le \frac{8}{\binom{q}{k}}.\end{equation}
    \end{proof}

    \begin{proof}[Proof of Lemma~\ref{lem: bound for k = q, EF}]
    Let $L \subseteq \items$ be the set of the bottom $\ceil{3n\log m} + (n - 1)$ items according to $\pi_i$. Notice that $\ceil{3n\log m} + (n - 1) \le 3n\log m + n \le 4n\log m$ since $\log(m )\ge \log(2n) \ge \log(4) \ge 1$.\footnote{We always interpret $\log$ as the natural log.} Recall that $b_{iq}$ is the worst item in agent $i$'s bundle. We consider the event that $b_{iq} \notin L$. 
    We show that the probability of $b_{iq} \in L$ is small and, conditioned on $b_{iq} \notin L$, it is unlikely that $i$ envies anyone. Formally, notice that
    \begin{align*}
        \Prs{\bigcup_{j : j < i} \cE^{ij}_q}
        &=\Prs{\bigcup_{j : j < i} \cE^{ij}_q \suchthat b_{iq} \notin L} \cdot \Prs{b_{iq} \notin L} + \Prs{\bigcup_{j : j < i} \cE^{ij}_q \suchthat b_{iq} \in L} \cdot \Prs{b_{iq} \in L} \\
        &\le \Prs{\bigcup_{j : j < i} \cE^{ij}_q \suchthat b_{iq} \notin L} + \Prs{b_{iq} \in L}\\
        &\le \sum_{j : j < i} \Prs{\cE^{ij}_q \suchthat b_{iq} \notin L} + \Prs{b_{iq} \in L}.
    \end{align*}
    We upper bound each of these terms individually.

    \begin{claim}\label{claim: bound on not L}
        For all $j < i$, $\Prs{\cE^{ij}_q \suchthat b_{iq} \notin L} \leq 1/m$.
    \end{claim}

    \begin{claim}\label{claim:bound on L}
    $\Prs{b_{iq} \in L} \leq 2\left(\frac{8n\log m}{m^{1 - 1/n}}\right)^{n + 1 - i}$.
    \end{claim}

    Given these two bounds, we can conclude the proof of the lemma as follows:
    \begin{equation}\label{eq:part2}
    	\Prs{\bigcup_{j: j < i} \cE^{ij}_q} 
     \le \sum_{j : j < i} \Prs{\cE^{ij}_q \suchthat b_{iq} \notin L} + \Prs{b_{iq} \in L}
     \le
     \frac{n}{m} + 2\left(\frac{8n\log m}{m^{1 - 1/n}}\right)^{n + 1 - i}.
    \end{equation}
    It remains to prove the two claims.

    \begin{proof}[Proof of Claim~\ref{claim: bound on not L}]
    Notice that $b_{iq} \notin L$ implies $A_i \cap L = \emptyset$, so all the items in $L$ must have been chosen by other agents.
    If $A_j \cap L \ne \emptyset$ it follows that   $b_{jq} \in L$ and, hence,  $b_{iq} \succ_i b_{jq}$.
    Therefore, $\Prs{\cE^{ij}_q \suchthat b_{iq} \notin L} \le \Prs{A_j \cap L = \emptyset \suchthat b_{iq} \notin L}$ and it suffices to upper bound the latter. 

    To analyze the bundle of an agent $j \neq i$, we revert to the second view of the allocation process in which the bundles $R_0, \ldots, R_q$ are first sampled and then matched to agents different that $i$. 
    We condition on any valid $R_0, \ldots, R_q$ such that $L \subseteq \bigcup_\ell R_\ell$, and now consider sampling the actual bundles $A_{j'}$ for $j' \ne i$ from these. Except for $R_0$ and $R_q$, each $|R_\ell| = n -1$ and $A_j$ contains exactly one item from $R_\ell$ chosen uniformly at random. Therefore, with probability $\frac{|R_{\ell} \cap L|}{n - 1}$, $A_j \cap R_{\ell} \cap L \ne \emptyset$. 
    In words, if there are $r$ items from $L$ in $R_\ell$, then $j$   receives one of them with probability $r/(n - 1)$. Importantly, these are independent across rounds $\ell$. It follows that
    \begin{align*}
    	\Prs{A_j \cap L = \emptyset \suchthat b_{iq} \notin L}
    	&\le \Prs{\sum_{\ell = 1}^{q-1} \mathbb{I}[A_j \cap L \cap R_\ell = \emptyset] = 0 }. 
    \end{align*}
    By  linearity of expectation and the choice of $L$, 
    \begin{equation}\label{eq:l-lower}
    	\E \left[\sum_{\ell = 1}^{q-1} \mathbb{I}[A_j \cap L \cap R_\ell] \right] = \sum_{\ell = 1}^{q - 1} \frac{|R_\ell \cap L|}{n - 1} \ge \frac{|L| - |R_0| - |R_q|}{n - 1} \ge \frac{|L| - (n - 1)}{n - 1} \ge \frac{3n \log m}{n - 1} \ge 3 \log m.
    \end{equation}
    The sum $\sum_{\ell = 1}^{q-1} \mathbb{I}[A_j \cap L \cap R_\ell = \emptyset]$ is the sum of independent Bernoulli variables indicating whether $j$ received an item from $L$ in round $\ell$. Applying a Chernoff bound we have
    \begin{align*}
    	\Prs{\sum_{\ell = 1}^{q-1} \mathbb{I}[A_j \cap L \cap R_\ell = \emptyset] = 0 }
    	&\le \Prs{\sum_{\ell = 1}^{q-1} \mathbb{I}[A_j \cap L \cap R_\ell = \emptyset] < \left(1 - \sqrt{\frac{2}{3}}\right)\E \left[\sum_{\ell = 1}^{q-1} \mathbb{I}[A_j \cap L \cap R_\ell] \right]}\\
    	&\le \exp\left(\frac{-\left(\sqrt{\frac{2}{3}}\right)^2 \E \left[\sum_{\ell = 1}^{q-1} \mathbb{I}[A_j \cap L \cap R_\ell ] \right]}{2}\right)\\
    	&\le \exp\left(\frac{-\frac{2}{3} (3 \log m)}{2}\right)\\
    	&= \exp(-\log m) = \frac{1}{m}.
    \end{align*}
    Putting this all together, we have that 
    \begin{equation}
        \Prs{\cE^{ij}_q \suchthat b_{iq} \notin L} \le \Prs{A_j \cap L = \emptyset \suchthat b_{iq} \notin L} \le 1/m, \label{eq:preijq}
    \end{equation}
    which concludes the proof of Claim~\ref{claim: bound on not L}.
    \end{proof}

    The proof of Claim~\ref{claim:bound on L} is deferred to Appendix~\ref{app: missing proofs}. This concludes the proof of Lemma~\ref{lem: bound for k = q, EF}.
    
\end{proof}

\subsection{Round-Robin is $\SDEFX$: The proof of Lemma~\ref{lem: rr is sdefx}}\label{subsec: RR is EFX}

We extend the analysis of~\Cref{subsec: RR EF} to get $\SDEFX$ allocations for arbitrary $m$.  Fix an arbitrary $m$ and let $m = qn + r$ for $1 \le r \le n$ and $r,q\in \mathbb{N}$. Let $\alloc^{q}$ be the allocation after $qn$ steps, and   $\alloc$ be the  complete allocation of all $qn+r$ items. Also, let $P_t$ denote the pool of available items at ``time'' $t$, i.e., after $t$ picks have been made (so, $P_0 = \items$ and $P_m = \emptyset$). Notice that
    \begin{align*}
    \Prs{\alloc \text{ is SD-EFX}}
    &\ge \Prs{\alloc \text{ is SD-EFX and } \alloc^q \text{ is SD-EF}}\\
    &= \Prs{\alloc \text{ is SD-EFX} \suchthat \alloc^q \text{ is SD-EF}} \cdot \Prs{\alloc^q \text{ is SD-EF}}.
    \end{align*}
    We   show that $\Prs{\alloc \text{ is SD-EFX} \suchthat \alloc^q \text{ is SD-EF}} = 1$ and $\Prs{\alloc^q \text{ is SD-EF}} \ge 1 - O\left(\frac{n^2}{m}\right)$. 
    
    We begin with the former.    
    Fix arbitrary samples $\pi_1, \ldots, \pi_n$ such that $\alloc^q$ is SD-EF. We want to show that $\alloc$ is SD-EFX. Consider agents $i, j \in \agents$. We show that either $A^q_i \sdpref_{\pi_i} A_j$ or, for all $g \in A_j$, $A^q_i \sdpref_{\pi_i} A_j \setminus \set{g}$.
    Since $A^q_i \subseteq A_i$, the same holds for $A_i$.
    It must be that either $A_j = A^q_j$ or $A_j = A^q_j \cup \set{g}$ for some $g \in P_{qn}$. If $A_j = A^q_j$ it follows that $A^q_i \succeq A^q_j = A_j$, since $\alloc^q$ is $\SDEF$ by assumption. 
    Suppose instead that $A_j = A^q_j \cup \set{g^*}$ for some $g^* \in P_{qn}$.  
    The key observation is that  $g^* \in P_{qn}$, implying it was available (and not selected) at every time step that $i$ picked an item, so  $g' \succ g^*$  for all $g' \in A_i$. Note that, this doesn't immediately imply the EFX condition, since $g^*$ might not be the worst item (from $i$'s perspective) that $j$ owns.
    Let $b_{ik}$ and $b_{jk}$ denote the $k$'th best item in $A^q_i$ and $A^q_j$ according to $\pi_i$, respectively. By assumption $b_{ik} \succ_i b_{jk}$ for all $k$. Let $b'_{jk}$ be the $k$'th best item in $A_j \setminus \set{g}$. 
    If $g^* = g$ then $A^q_i \sdpref_{\pi_i} A_j \setminus \set{g^*} = A^q_j$. 
    Otherwise, if $g^* \ne g$, suppose $g^* = b'_{j\ell}$ for some $\ell$. Notice that for all $k < \ell$, $b'_{jk} = b_{jk}$ so $b_{ik} \succ_i b'_{jk}$. For $k \ge \ell$, we have that $b_{ik} \succ_i g^* = b_{j\ell} \succeq_i b'_{jk}$. Hence, $A^q_j \sdpref_{\pi_i} A_j \setminus \set{g^*}$, as needed.
    
    Next, we show $\Prs{\alloc^q \text{ is SD-EF}} \ge 1 - O\left(\frac{n^2}{m}\right)$. Specifically, we will show that the probability that $\alloc^q$ is not SD-EF is upper bounded by $\frac{2066n^2}{m}$. 
    We assume that $m \ge 2000n^2$ (otherwise the bound is trivial), which implies that $qn \ge m/2$. The analysis is very similar to showing Round Robin is SD-EF when $m = qn$, so we simply describe which changes need to be made. Let $b^i_{jk}$ denote the $k$'th best item according to $\pi_i$ in the bundle $A^q_j$; for ease of notation we write $b_{jk}$ when $i$ is clear from the context. Let $\cE^{ij}_k$ be the that $b_{jk} \succ_i b_{ik}$, and let $\cE^{ij} = \bigcup_{k = 1}^q \cE^{ij}_k$ be $A^q_i \notsdpref_{\pi_i} A^q_j$.
    That is, $\cE^{ij}$, $b^i_{jk}$, and $\cE^{ij}_k$ are defined analogously to the proof of Lemma~\ref{lem: round robin is ef}, but with respect to $\alloc^q$ instead of $\alloc$.
    Let $R_\ell$ be the set of items picked between agent $i$'s $\ell$'th and $(\ell + 1)$'th pick, and 
    $R_q = P_{qn}$ be the set of items remaining items after the allocations in  $\alloc^q$.

    On a high level, our goal is again to upper bound
    by getting bounds on $\Prs{\cE^{ij}_k}$ and $\sum_{i \in \agents} \Prs{\bigcup_{j: j < i} \cE^{ij}_q}$, similarly to Lemmas~\ref{lem: bound for k < q, EF} and~\ref{lem: bound for k = q, EF}. Using very similar bounds we have

    \begin{align*}
        \Prs{\bigcup_{i,j} \cE^{ij} }
        &\le \sum_{i \in \agents} \sum_{k = 1}^q \Prs{\bigcup_{j: j < i} \cE^{ij}_k}\\
        &\le \sum_{i \in \agents} \left(\sum_{k = 1}^{q - 1}\Prs{\bigcup_{j} \cE^{ij}_k } + \Prs{\bigcup_{j} \cE^{ij}_q }  \right)\\
        &\le n \sum_{k = 1}^{q - 1}\Prs{\bigcup_{j} \cE^{ij}_k } + \sum_{i \in \agents} \Prs{\bigcup_{j} \cE^{ij}_q } \\
        &\le 8n \sum_{k = 1}^{q - 1}\frac{1}{\binom{q}{k}} + \sum_{i \in \agents} \Prs{\bigcup_{j} \cE^{ij}_q }  .
    \end{align*}

Since $qn \ge m/2$ we can, similarly to Inequality~\eqref{eq:q-sum}, bound the first term by $\frac{64n^2}{m}$. For the second term, one complication is that our bound on $\Prs{\bigcup_{j} \cE^{ij}_q }$ needs to take care of the $n=2$ case, separately. This is caused because the analog of Lemma~\ref{lem: bound for k = q, EF} is not a bound of $\frac{n}{m} + 2\left(\frac{8n\log m}{m^{1 - 1/n}}\right)^{n + 1 - i}$, but a bound of $\frac{n}{m} + 2\left(\frac{10n\log m}{m^{1 - 1/n}}\right)^{n + 1 - i + r}$, which is too weak for $n=2$.
We explain these differences, and how to address them, in Appendix~\ref{app: missing from rr efx}.

\subsection{Round-Robin is asymptotically optimal: The proof of Lemma~\ref{lem: rr lower bound}}\label{subsec: RR is asymptotically optimal}


Fix $n$ and $m$ with $m \ge 2n$. We prove that $\EF$ and $\SDEFX$ allocations do not exist with probability at least $\min\left(\frac{n^2}{8m}, \frac{1}{8}\right)$. For $\SDEFX$ we first show that if two agents have the same favorite item, no $\SDEFX$ allocation exists. For $\EF$, if two agents have unit demand valuations, then clearly having the same favorite item also implies that no $\EF$ allocations exist. We then show using standard birthday paradox bounds that the probability two agents have the same favorite item is at least $\min\left(\frac{n^2}{8m}, \frac{1}{8}\right)$.

For the former, suppose two agents have the same favorite item $g$ with $m \ge 2n$. We show that no $\SDEFX$ allocation exists. Fix an allocation $\alloc$. First note that if $\alloc$ is not \emph{balanced}, i.e., $||A_i| - |A_j|| \le 1$, then it cannot be SD-EFX. Indeed, if $|A_i| - |A_j| > 1$, then for any $g \in A_i$, $A_i \setminus \set{g} > |A_j|$ so it cannot be the case that $|A_j| \sdpref_{\pi_j} A_i \setminus \set{g}$. Next, suppose $\alloc$ is balanced. Since the two agents have the same favorite item, there is some agent $i$ whose favorite good is $g$ and who did not receive $g$. Further, there is an agent $j$ such that $g \in A_j$, and since $\alloc$ is balanced, $|A_j| \ge 2$. Therefore, there is some $g' \ne g$ such that $g' \in A_j$. However, since $g \in A_j \setminus \set{g'}$, $A_i \notsdpref_{\pi_i} A_j \setminus \set{g'}$, so $\alloc$ is not SD-EFX.

Next, we show that the probability two agents have the same favorite item is at least $\min\left(\frac{n^2}{8m}, \frac{1}{8}\right)$. Notice that each agent's favorite item is a uniformly selected item, independent of other agents. The probability that two agents have the same first item is known from the collision analysis of Hash functions. Indeed, it is known as long as $n \le \sqrt{2m}$, the probability that two agents have the same first item is at least $\frac{n(n - 1)}{4m}$~\cite{katz2020introduction}. In other words, as long as $m \ge n^2/2$, two agents have the same first item with probability at least $\frac{n(n - 1)}{4m} \ge \frac{n^2}{8m}$. For $m < n^2/2$, notice that the probability of a collision is decreasing in $m$. Hence, the probability of a collision is at least the probability with $n$ agents and $m' = \ceil{n^2/2} \le n^2$ items. The previous result says that this probability is then at least $\frac{n^2}{8m'} \ge \frac{1}{8}$.

\subsection{Give-Away Round Robin}\label{subsec: give away}

  As we've shown, Round-Robin is asymptotically optimal as $n$ grows. However, while it is often reasonable to assume that the number of items is large, it is often less reasonable to assume the same for the number of agents. 
  When $n$ is a constant, no $\SDEF$ allocations exist with probability $\Omega_n(1/m)$. Although Round-Robin finds $\SDEF$ allocations with probability $1 - O_n(1/m)$, when $m$ is divisible by $n$, \Cref{thm:positive for order consistent} only gives an upper bound of $1 - O_n(\frac{\log m}{m^{1 - 1/n}})$ on the probability of finding such an allocation.
  A careful analysis shows this is relatively tight: with probability $\Omega_n(\frac{1}{m^{1 - 1/n}})$, the last agent will be left with their worst item, which implies the resulting allocation cannot be $\SDEF$. For example, with $n = 2$, Round-Robin will not find an $\SDEF$ allocation with probability $\Omega(1/\sqrt{m})$, while the lower bound only implies such allocations don't exist with probability $\Omega(1/m)$. This motivates an alternative algorithm that closes this gap. To avoid receiving a truly bad item, agents first \emph{give away} a bad item to each of the other agents, then proceeds with  regular Round-Robin on the remaining items. We formalize this in Algorithm~\ref{alg:give-away}, and prove it achieves optimal (up to log factors) probability of finding an $\SDEF$ allocation.

\begin{algorithm}\label{alg:give-away}
\caption{Give-Away Round Robin}
\DontPrintSemicolon
    \textbf{Input:} Valuations $v_1, \ldots, v_n$ which are order consistent with respect to $\pi_1, \ldots, \pi_n$.\\
    \textbf{Output:} An allocation $\alloc = (A_1, A_2, \ldots, A_n)$.\\
    \hrulefill

    \textbf{Set} $A_i \gets \emptyset$ for all agents $i \in \agents$\;
\textbf{Set} $P \gets \items$\;

    \textbf{Phase 1: Give-away} ($n(n - 1)$ items)

    \For{$i = 1, \ldots, n$}{
        \For{$j = 1, \ldots, n$, $j \ne i$}{
            {\bf Let} $g \in P$ be the unallocated item with the highest index in $\pi_i$.\;
            {\bf Set} $A_j \gets A_j \cup \{g\}$\;
            {\bf Set} $P \gets P \setminus \{g\}$\;
        }
        
    }
    
    \textbf{Phase 2: Round-Robin}
    
    \For{$i = 1, \ldots, n, 1, \ldots, n, 1, \ldots$}{
    {\bf Let} $g \in P$ be the unallocated item with the lowest index in $\pi_i$.\;
    {\bf Set} $A_i \gets A_i \cup \{g\}$\;
    {\bf Set} $P \gets P \setminus \{g\}$\;
}
\Return{$(A_1, A_2, \ldots, A_n)$}

\end{algorithm}

\begin{theorem}\label{thm:give-away}
    Let $\pi_1, \ldots, \pi_n$ be sampled independently and uniformly at random. When  $m = qn$, $q\in \mathbb{N}$,
    \[
        \Prs{\text{Give-Away Round Robin is SD-EF}} \ge 1 - \tilde{O}_n\left(\frac{1}{m} \right).
    \]
\end{theorem}
\begin{proof}
The proof follows a very similar structure to the proof of \Cref{thm:positive for order consistent}, though with more intricate analysis in various places. Fix $m = qn$ for an integer $q$. We will show the probability that Give-Away Round Robin is not SD-EF is upper bounded by
\[
    \frac{67584n^3\log^2 m + 25216n^4\log m + 2354n^5 + 128n^2 \log m + 24n^3 + 3n^2 + n}{m}. 
\]
Throughout, we will assume that $m$ is lower bounded by the numerator, as otherwise, the bound is trivial. 

For $i \ne j$, we again let $\cE^{ij}$ be the event that $A_i \notsdpref_{\pi_i} A_j$.
Unlike in Round Robin,  the give-away phase makes it  possible that lower-indexed agents envy higher-indexed agents, so $\cE^{ij}$ is defined for all  agents $i \ne j$. 
We again wish to upper bound $\Prs{\bigcup_{i \in \agents}\bigcup_{j: j \ne i} \cE^{ij}}$. We reuse the notation $g \succ_i g'$ to denote preferences under $\pi_i$, and $b^i_{jk}$ for $i$'s $k$'th favorite item in $A_j$. We   decompose $\cE^{ij}$ into $\bigcup_{k = 1}^q \cE^{ij}_k$ as before,  where $\cE^{ij}_k$ is the event that $b_{jk} \succ_i b_{ik}$. Using a union bound, we wish to upper bound
\[
    \sum_{i \in \agents} \Prs{\bigcup_{k \le q} \bigcup_{j: j \ne i} \cE^{ij}_k}.
\]

We let $\gamma_{ij}$ for $i \ne j$ be the item given from $i$ to $j$ in the give-away phase. We use $\tilde{P}_t$ for $0 \le t \le n(n - 1)$ to denote the pool of available items after $t$ giveaway steps. We then let $P_t$ be the pool after $t$ Round Robin steps (so $\tilde{P}_0 = \items$, $\tilde{P}_{n(n-1)} = P_0$, and $P_{m - n(n - 1)} = \emptyset$). We let $g_t$ be the item taken at the $t$'th step of Round Robin. Therefore, $A_j = \set{\gamma_{ij} \suchthat i \ne j} \cup \set{g_t \suchthat t \equiv j \bmod n}$.

We fix an agent $i$ and condition on $\pi_i$. The random process of selecting items is now similar to in Round Robin. For times $t$ when $\gamma_{ij}$ is given, $\gamma_{ij}$ is the worst remaining item from $\tilde{P}_{t-1}$, and for $\gamma_{i'j}$ with $i' \ne i$ it is a uniformly selected item from $\tilde{P}_{t-1}$ (from  agent $i$'s perspective). 
During the  Round Robin phase, for $g_t$ with $t \equiv i \bmod n$, $g_t$ is the best available item from $P_{t - 1}$, and when $t \not\equiv i \bmod n$, $g_t$ is a uniformly selected item from $P_{t - 1}$. For the Round Robin phase, we again consider the sets $R_\ell$ of items picked in the Round Robin phase between $i$'s $\ell$'th and $(\ell + 1)$'th pick, so $R_0 = \set{g_1, \ldots, g_{i - 1}}$, for $1 \le \ell \le q - n$, $R_\ell = \set{g_{n (\ell - 1) + i + 1}, \ldots, g_{\ell n + i - 1}}$, and $R_q = \set{g_{n(q - 1) + i+1}, \ldots, g_m}$.  As with Round Robin, each item in  $R_\ell$ is picked sequentially and uniformly at random from the remaining items. This allows us to analyze an equivalent two-step process wherein $R_\ell$ is first selected uniformly at random from $P_t$, after which each item in $R_\ell$ is matched to a (non-$i$) agent that is picked in the corresponding timesteps uniformly at random. 

We now break down $\bigcup_{k \le q}$ via a union bound into different groupings of $k$. 
\begin{align*}
    \sum_{i \in \agents} \Prs{\bigcup_{k \le q} \bigcup_{j: j \ne i} \cE^{ij}_k} =  \sum_{i \in \agents}\Bigg(& \Prs{\bigcup_{k: k < 32\log m + 4n} \bigcup_{j: j \ne i} \cE^{ij}_k} + \Prs{\bigcup_{k: 32\log m + 4n \le k \le q - 32\log m - 5n} \bigcup_{j: j \ne i} \cE^{ij}_k}\\
    &\qquad + \Prs{\bigcup_{k: q - 32\log m - 5n <  k \le q - 1 } \bigcup_{j: j \ne i} \cE^{ij}_k} + \Prs{\bigcup_{k: k = q} \bigcup_{j: j \ne i} \cE^{ij}_k}\Bigg).
\end{align*}
Fix an agent $i$.  The following lemmas  handle each of the summands individually. 
\begin{lemma}\label{lem:sum1}
    $\Prs{\bigcup_{k: k < 32\log m + 4n} \bigcup_{j: j \ne i} \cE^{ij}_k} \leq \frac{2048n^2 \log^2 m + 640n^3\log m + 50n^4}{m}$.
\end{lemma} 

\begin{lemma}\label{lem:sum2}
    $\Prs{\bigcup_{k: 32\log m + 4n \le k \le q - 32\log m - 5n} \bigcup_{j: j \ne i} \cE^{ij}_k} \leq \frac{1}{m} $.
\end{lemma} 
\begin{lemma}\label{lem:sum3}
    $\Prs{\bigcup_{k: q - 32\log m - 5n <  k \le q - 1 } \bigcup_{j: j \ne i} \cE^{ij}_k} \leq \frac{65536 n^2\log^2 m + 24576 n^3 \log m + 2304 n^4 + 128n \log m + 24n^2 + n}{m}$.
\end{lemma} 
\begin{lemma}\label{lem:sum4}
   $\Prs{\bigcup_{k: k = q} \bigcup_{j: j \ne i} \cE^{ij}_k} \leq \frac{2n}{m}$.
\end{lemma} 

Together, these imply that 
\begin{align*}
    &\sum_{i \in \agents} \Prs{\bigcup_{k \le q} \bigcup_{j: j \ne i} \cE^{ij}_k}\\
    &\qquad\leq n\left(\frac{67584n^2\log^2 m + 25216n^3\log m + 2354n^4 + 128n \log m + 24n^2 + 3n + 1}{m} \right)\\
    &\qquad \le \frac{67584n^3\log^2 m + 25216n^4\log m + 2354n^5 + 128n^2 \log m + 24n^3 + 3n^2 + n}{m} \in \tilde{O}_n\left(\frac{1}{m}\right). \qedhere
\end{align*}
\end{proof}

We prove \Cref{lem:sum1,lem:sum2,lem:sum3,lem:sum4} in Appendix~\ref{app: missing from give away}.

\section{Conclusion}

In this paper, we study the existence of envy-free allocations beyond the case of additive valuations.
We introduce a simple model to study this question: starting from a worst-case valuation function, randomly rename the items. We show that, in this model, if valuations are order-consistent (a valuation class general enough to include additive, unit-demand, budget additive, single-minded, etc), even sd-envy-free allocations exist with high probability. In fact, a simple Round-Robin process will output such an allocation with high probability. Our bound on the probability is asymptotically tight, but we can improve upon it for the important case of a constant number of agents, using a variation of Round-Robin, that might be of independent interest. For arbitrary valuations, we show a positive result for the case of $n=2$ agents. Our proof reduces the question about the existence of envy-free allocations to a question about the number of automorphisms in a uniform hypergraph.

An important problem we leave open is whether envy-free allocations exist for general valuation functions for $n > 2$ agents. 
Our approach of viewing valuations as hypergraphs immediately fails since now an allocation is not just a set and its complement. Interestingly, the existence of envy-free allocations in our model does not seem to become any easier, even if one is willing to make fairly common structural assumptions on the valuation functions. For example, one might attempt to prove~\Cref{theorem:ef-whp-gen} for, e.g., monotone submodular functions (where strong concentration results are readily available~\cite{vondrak2010note}), hoping to get a proof for $n=2$ that can extend to more agents. In Section~\ref{sec: general} we show that corresponding to a valuation function, we can construct a unique hypergraph; we complement this in  \Cref{appendix:submod-hgraph} by showing that for any  $\hgraph$ we can   construct a monotone submodular function $v$ such that $\hgraph^v = \hgraph$. 
This implies that, in a sense, the  general existence of balanced EF allocations for two agents reduces to the existence of EF allocations for  
monotone submodular valuation functions.

\bibliographystyle{alpha}
\bibliography{refs.bib}

\appendix
\section{Preference Hypergraphs and Submodular valuations}\label{appendix:submod-hgraph}

We conclude by showing that the existence of balanced envy-free allocations for arbitrary valuation functions can, in a sense, be reduced to the class of submodular valuation functions. 

\begin{proposition}
For any valuation function $v$ and corresponding hypergraph $\hgraph^v  = (\items, E)$ such that, for all $S \subseteq \items$ of size $m/2$ exactly one of $S$ or $\overline{S}$ is in $E$, there exists a monotone, submodular valuation function $v'$ so that $\hgraph^v = \hgraph^{v'}$. 
\end{proposition}
\begin{proof}
    Construct $v'$ as follows: 
    \[v'(S) = \begin{cases}
    |S| & \text{ if } |S| < m/2\\
    m/2 - 1/2 & \text{ if } |S| = m/2 \text{ and } \overline{S} \in E \\
    m/2 - 1/4 & \text{ if } |S| = m/2 \text{ and } S \in E  \\ 
    m/2 & \text{ if } |S| > m/2.
\end{cases}\]
It is straightforward to verify that $v'$ is submodular and that $\hgraph^{v'} = \hgraph^{v}$. 
\end{proof}

\section{Omitted proofs}\label{app: missing proofs}

    \begin{proof}[Proof of Claim~\ref{claim: gamma stuff 1}] 
    Using Gautchi's inequality we have that
    \begin{align*}
        \frac{2}{n} \cdot \frac{\Gamma(k + 1/n)\Gamma(q - k + 1/n)}{\Gamma(q + 1/n)}
        &= \frac{2}{n} \cdot \frac{\Gamma(k + 1)\Gamma(q - k +1)}{\Gamma(q + 1)}\cdot \frac{\Gamma(k + 1/n)}{\Gamma(k + 1)}\cdot \frac{\Gamma(q - k + 1/n)}{\Gamma(q - k + 1)}\cdot \frac{\Gamma(q + 1)}{\Gamma(q + 1/n)}\\
        &\le \frac{2}{n} \cdot \frac{k!(q-k)!}{q!} \frac{1}{k^{1 - 1/n}}\frac{1}{(q - k)^{1 - 1/n}}(q + 1)^{1 - 1/n}\\
        &= \frac{2}{n}\frac{1}{\binom{q}{k}} \left(\frac{q + 1}{k(q - k)}\right)^{1 - 1/n}\\
        &< \frac{2}{n}\frac{1}{\binom{q}{k}} \left(\frac{2q}{q/2}\right)^{1 - 1/n}\\
        &\le\frac{2}{n}\frac{1}{\binom{q}{k}} 4^{1 - 1/n}\\
        &\le \frac{8}{n} \cdot \frac{1}{\binom{q}{k}}
    \end{align*}
    where the third inequality holds because $q + 1 \le 2q$, both $k \ge 1, (q - k) \ge 1$, and at least one $k$ and $(q - k)$ is $\ge q/2$. 
    \end{proof}


    \begin{proof}[Proof of Claim~\ref{claim:bound on L}]
    Recall that $b_{iq} = g_{n(q - 1) + i}$. Further, $g_{n(q - 1) + i} \in L$ requires that $P_{n(q - 1) + i - 1} \subseteq L$,  otherwise an item $g \notin L$, which $i$ prefers to any item in $L$, is available (and $i$ would have preferred to choose it).
    Since $|P_{n(q - 1) + i - 1}| = m  - (n(q - 1) + i - 1) =  n - i + 1$, we have
    \[
    	\Prs{b_{iq} \in L} = \Prs{P_{n(q - 1) + i - 1} \subseteq L} = \sum_{S \in \binom{L}{n + 1 - i}} \Prs{P_{n(q - 1) - i + 1} = S}.
    \]
    Fixing an arbitrary $S \in \binom{L}{n + 1 - i}$ we have
    \begin{align*}
    	\Prs{P_{n(q - 1) - i + 1} = S}
    	&= \Prs{\bigcap_{t = 1}^{n(q - 1) + i - 1} \set{g_t\notin S} }\\
		&= \prod_{t = 1}^{n(q - 1) + i - 1} \Prs{g_t \notin S \suchthat \bigcap_{t' = 1}^{t - 1} g_{t'} \notin S }\\
		&= \prod_{t = 1}^{n(q - 1) + i - 1} \Prs{g_t \notin S \suchthat S \subseteq P_{t - 1} }.
    \end{align*}
    For timesteps $t$ where $i$ made a pick, i.e., $t \equiv i \pmod{n}$, $\Prs{g_t \notin S \suchthat S \subseteq P_{t - 1} }$ is hard to compute, as it depends on exactly which items are available. In these cases, we use the trivial upper bound of $1$. For other timesteps $t$,   $g_t$ is a uniformly random item from $P_{t - 1}$, so $\Prs{g_t \notin S \suchthat S \subseteq P_{t - 1} } = 1 - \frac{|S|}{|P_{t - 1}|} = 1 - \frac{n + 1 - i}{m + 1 - t} = \frac{m + 1 - t - (n + 1 - i)}{m + 1 - t}$. It follows that
    \[
    	\Prs{P_{n(q - 1) - i + 1} = S} \leq \prod_{t = 1: t \not\equiv i \bmod n}^{n(q - 1) + i - 1} \frac{m + 1 - t - (n + 1 - i)}{m + 1 - t}.
    \]
    To ease notation, we perform a change of variables to $t' = n(q - 1) + i - t$, i.e., we reverse the product. The condition $t \not\equiv i \pmod n$ becomes  $t' \not \equiv 0 \pmod{n}$, and we have $t' = n(q - 1) + i - t = m  + 1 - t - (n + 1 - i)$. Therefore, we can simplify to
    \begin{align*}
    	\prod_{t' = 1: t' \not\equiv 0 (\bmod n)}^{n(q - 1) + i - 1} \frac{t'}{t' + (n + 1 - i)}
    	&= \frac{\prod_{t' = 1}^{n(q - 1) + i - 1} \frac{t'}{t' + (n + 1 - i)}}{\prod_{t' = 1: t' \equiv 0 (\bmod n)}^{n(q - 1) + i - 1} \frac{t'}{t' + (n + 1 - i)}}\\
    	&= \prod_{t' = 1}^{n(q - 1) + i - 1} \frac{t'}{t' + (n + 1 - i)} \cdot \prod_{t' = 1: t' \equiv 0 (\bmod n)}^{n(q - 1) + i - 1} \frac{t' + (n + 1 - i)}{t'}.
    \end{align*}
    Observing that $n(q - 1) + i - 1 + (n + 1 - i) = qn = m $, the first product is exactly equal to
    
    \[\frac{(m - (n + 1 - i))!(n + 1 - i)!}{m!} = \frac{1}{\binom{m}{n + 1 - i}} \le \frac{1}{(m - (n + 1 - i))^{n + 1 - i}} \le \frac{1}{(m - n)^{n + 1 - i}} \leq^{(m \ge 2n)} \left(\frac{2}{m}\right)^{n + 1 - i}. \]
    
     The second product simplifies, as we only need to consider values $t' = \ell n$ for $1 \le \ell \le q - 1$. Hence,
    \[
    	\prod_{t' = 1: t' = 0 \bmod n}^{n(q - 1) + i - 1} \frac{t' + (n + 1 - i)}{t'} = \prod_{\ell = 1}^{q - 1} \frac{\ell n + (n + 1 - i)}{\ell n} = \prod_{\ell = 1}^{q - 1} \frac{\ell + \frac{n + 1 - i}{n}}{\ell} = \frac{(1 + \frac{n + 1 - i}{n})^{\overline{q - 1}}}{(q - 1)!},
    \]
    where the notation $(x)^{\overline{r}}$ represents the rising factorial $(x)^{\overline{r}} = x(x + 1) \cdots (x + r - 1)$. We use the fact that $(x)^{\overline{r}} = \frac{\Gamma(x + r)}{\Gamma(x)}$ to get that the above is equal to,
    \[
    	\frac{\Gamma(q + \frac{n + 1 - i}{n})}{\Gamma(1 + \frac{n + 1 - i}{n})\Gamma(q)}.
    \]
    Similarly to Claim~\ref{claim: gamma stuff 1}, using the fact that the Gamma function is lower bounded by $1/2$ and Gautchi's inequality~\cite{gautschi1959some}, along with the assumption that $m \ge 2n$,  the above is at most
    \[\frac{2\Gamma(q + \frac{n + 1 - i}{n})}{\Gamma(q)} \le 2\left(q + \frac{n + 1 - i}{n}\right)^{\frac{n + 1 - i}{n}} \le 2\left(q + 1\right)^{\frac{n + 1 - i}{n}} \le 2m^{\frac{n + 1 - i}{n}} = 2(m^{1/n})^{n + 1 - i}.
    \]
    Combining with the previous bounds, we obtain for any $S \in \binom{L}{n + 1 - i}$ that, $\Prs{P_{n(q - 1) - i + 1} = S} \le
        \left(\frac{2}{m}\right)^{n + 1 - i}\cdot   2(m^{1/n})^{n + 1 - i}=  2\left(\frac{2}{m^{1 - 1/n}}\right)^{n + 1 - i}$. 
    The number of possible choices of $S$ is
    \begin{equation}\label{eq:num-s-bound}
    	\binom{|L|}{n + 1 - i} \le |L|^{n + 1 - i} \le (4n\log m)^{n + 1 - i}.
    \end{equation}

    Union bounding over all of these choices, we get that
    \[
        \Prs{b_{iq} \in L} =  \Prs{P_{n(q - 1) - i + 1} \subseteq L} \le 2\left(\frac{8n\log m}{m^{1 - 1/n}}\right)^{n + 1 - i},
    \]
    which concludes the proof of Claim~\ref{claim:bound on L}.
    \end{proof}

\subsection{Missing details from the proof of Lemma~\ref{lem: rr is sdefx}}\label{app: missing from rr efx}

Here, we flesh out the missing details from the proof of Lemma~\ref{lem: rr is sdefx}.

    For $k < q$, the proof of Lemma~\ref{lem: bound for k < q, EF} goes through almost identically, to give us  
    $\Prs{\cE^{ij}_k} \le \frac{8}{\binom{q}{k}}$. The only difference is that the equality in $m - ((\ell - 1)n + j - 1) = (q - \ell)n + n - j + 1)$ of \Cref{eq:m-to-q} becomes a weak inequality; the inequality is in the right direction for the same conclusion to hold.
        
    For $k = q$, the proof of Lemma~\ref{lem: bound for k = q, EF}, the analysis is very similar, even though more care is needed.
    To ensure that there are at least $3\log m$ items remaining when lower bounding the expectation in Equation~\eqref{eq:l-lower}, since $|R_0| + |R_q|$ is only upper bounded by $(n - 1) + r \le 2n - 1$, we slightly change $L$, which is now defined as the bottom $\ceil{3n\log m} + 2n - 1$ items. It follows that $|L| \le 5 n \log m$. A second change is in the proof of Claim~\ref{claim:bound on L}. $|P_{n(q - 1) + i - 1}|$, the number of items remaining when $i$ makes their $q$'th pick, is now $n - i + 1 + r$. The analysis continues to go through by replacing occurrences of $n - i + 1$ with $n - i + 1 + r$. In the change of variables, we have $t' = n(q - 1) + i - t = m  + 1 - t - (n + 1 - i + r)$.
    This means that Inequality~\eqref{eq:num-s-bound} becomes
    \[
    	\binom{|L|}{n + 1 - i + r} \le |L|^{n + 1 - i + r} \le (5n\log m)^{n + 1 - i + r},
    \]
    and Inequality~\eqref{eq:part2} becomes
    \begin{equation}\label{eq:new-part2}
    	\Prs{\bigcup_{j: j < i} \cE^{ij}_q} \le \frac{n}{m} + 2\left(\frac{10n\log m}{m^{1 - 1/n}}\right)^{n + 1 - i + r}.
    \end{equation}
    
    For $n = 2$, we will need a bound that is even stronger than \eqref{eq:new-part2}. When there are two agents the only envy event we need to consider is $\cE^{21}_q$. We claim   that we need only include $i$'s worst  $1 + r \le 3$ items in $L$ to get 
    $
    	\Prs{\cE^{21}_q \suchthat b_{2q} \notin L} = 0.
    $    
  Indeed, suppose $b_{2q} \notin L$, so  $b_{2q}$ is not one of the $3$ worst items according to $\pi_2$. 
  Since $|R_q| \leq 2$
  and $L \cap A^q_2 = \emptyset$, it must be that $A^q_1 \cap L \ne \emptyset$. We conclude $b_{1q} \in L$, so $b_{2q} \succ_2 b_{1q}$ and event $\cE^{21}_q$ does not occur. Therefore, when $n = 2$,
  \begin{equation}\label{eq:new-part2-n2}
  	\sum_{i \in \agents} \Prs{\bigcup_{j: j < i} \cE^{ij}_q} = \Prs{\cE^{21}_q} \le 2\left(\frac{3}{m^{1 - 1/n}} \right)^{n + 1 - i + r} \le 2\left(\frac{2n}{m^{1/2}} \right)^{1 + r}
  \end{equation}

 We are now ready to upper bound the probability that $\alloc^q$ is not SD-EF. 	
\begin{align*}
        \Prs{\bigcup_{i,j} \cE^{ij} }
        &\le \sum_{i \in \agents} \sum_{k = 1}^q \Prs{\bigcup_{j: j < i} \cE^{ij}_k}\\
        &\le \sum_{i \in \agents} \left(\sum_{k = 1}^{q - 1}\Prs{\bigcup_{j} \cE^{ij}_k } + \Prs{\bigcup_{j} \cE^{ij}_q }  \right)\\
        &\le n \sum_{k = 1}^{q - 1}\Prs{\bigcup_{j} \cE^{ij}_k } + \sum_{i \in \agents} \Prs{\bigcup_{j} \cE^{ij}_q } \\
        &\le 8n \sum_{k = 1}^{q - 1}\frac{1}{\binom{q}{k}} + \sum_{i \in \agents} \Prs{\bigcup_{j} \cE^{ij}_q }  .
    \end{align*}
    Since $qn \ge m/2$ we can refine Inequality~\eqref{eq:q-sum} to get   
    $\sum_{k = 1}^{q - 1}\frac{1}{\binom{q}{k}} \le \frac{4}{q} \le \frac{8n}{m}$, so the first term is at most $\frac{64n^2}{m}$.
    
    When $n = 2$, the second term is bounded by $2\left(\frac{2n}{m^{1/2}} \right)^{1 + r}$ by \eqref{eq:new-part2-n2}. 
    By our assumption $m \ge 2000n^2$ and $r\geq 1$, so  $\frac{2n}{m^{1/2}} \le 1$ and
    \[
    2\left(\frac{2n}{m^{1/2}} \right)^{1 + r} \le 2\left(\frac{2n}{m^{1/2}} \right)^{1 + 1} = \frac{8n^2}{m}, 
    \]
    from which it follows that 
    $
    \Prs{\bigcup_{i,j} \cE^{ij} } \leq \frac{64n^2}{m} + \frac{8n^2}{m}
    = \frac{72n^2}{m} \le \frac{2066n^2}{m}.
    $
    
   Next, we show the same bound for $n\ge 3$.  Here
    \begin{align*}
    	\Prs{\bigcup_{j} \cE^{ij}_q }
    	&\le \frac{n^2}{m} + 2\sum_{i \in \agents} \left(\frac{10n\log m}{m^{1 - 1/n}}\right)^{n + 1 - i + r}\\
    	&\le \frac{n^2}{m} + 2\sum_{i =1}^\infty  \left(\frac{10n\log m}{m^{1 - 1/n}}\right)^{i + r}\\
    	&\le \frac{n^2}{m} + 2\sum_{i =1}^\infty  \left(\frac{10n\log m}{m^{2/3}}\right)^{i + r}.\\
    \end{align*}
    We   claim that
    \[
    	\frac{\log m}{m^{1/6}} \le \frac{6}{e}.
    \]
    To see why, observe that the derivative of $\frac{\log m}{m^{1/6}}$ is $\frac{6 - \log(m)}{6 m^(7/6))}$ the function is maximized over positive reals  exactly at $m = e^6$, where it takes value $\frac{6}{e}$. Substituting into the above, 
    
    \[
    \Prs{\bigcup_{j} \cE^{ij}_q } \le
     \frac{n^2}{m} + 2\sum_{i =1}^\infty  \left(\frac{60n}{e m^{1/2}}\right)^{i + r}.\]
    Since $(60/e)^2 \le 500$ and  we  are restricted to $m \ge 2000n^2$, $\frac{60n}{e m^{1/2}} \le \sqrt{\frac{500}{2000}}  = 1/2$. Using the fact that $r \ge 1$, this infinite series is at most
    \[\frac{n^2}{m} + 4 \left(\frac{60n}{e m^{1/2}}\right)^{2} \le \frac{2001n^2}{m}.\]
    Hence,
    \[\Prs{\bigcup_{j} \cE^{ij}_q } \le \frac{2066n^2}{m},\] as needed.

\subsection{Missing lemmas from Section~\ref{subsec: give away}}\label{app: missing from give away}

    \begin{proof}[Proof of \Cref{lem:sum1}]
    Let $T$ be the top $\floor{32 \log m + 4n}$ items according to $\pi_i$. Notice that if $T \subseteq A_i$, $b_{ik} \succ b_{jk}$ for all $k \le 32 \log m + 4n$. More formally, this means that,
    \[
        \Prs{\bigcup_{k: k < 32\log m + 4n} \bigcup_{j: j \ne i} \cE^{ij}_k} \le \Prs{T \nsubseteq A_i},
    \]
    so we directly upper bound the latter. For $T \nsubseteq A_i$ to occur, it must be the case that some other agent $j$ received an item in $T$ before $i$ made their $(\floor{32 \log m + 4n})$'th pick. Since we are assuming $m \ge 32\log m + 4n + n^2 \ge |T| + n^2$, $i$ would never give away an item in $T$ during the giveaway phase as a worse item must be available. Hence, the only way that another agent receives   an item in $T$ is if are given it by  a non-$i$ agent in the giveaway phase or they pick it in the Round Robin phase before time $n(|T| - 1) + i \le n|T|$. From $i$'s perspective all of these items are chosen uniformly at random from the pool. There are at most $(n - 1)(|T| + n) \le n|T| + n^2 \le 32n\log m + 5n^2$ such items. When they are taken, the pool is always of size at least $m - (32n\log m + 5n^2)$. Hence, each is   selected with probability at most $\frac{32n\log m + 5n^2}{m - 32n\log m - 5n^2}$. Since we are assuming $m \ge 64 n \log m + 10n^2$, this probability is at most $\frac{64n\log m + 10n^2}{m}$. Union bounding over the at most $32n\log m + 5n^2$ choices of $k$, this implies that the total probability is at most
    \begin{align*}
        \Prs{\bigcup_{k: k < 32\log m + 4n} \bigcup_{j: j \ne i} \cE^{ij}_k}
        &\le \frac{(64n\log m + 10n^2)(32n\log m + 5n^2)}{m}\\
        &= \frac{2048n^2 \log^2 m + 640n^3\log m + 50n^4}{m} \qedhere
    \end{align*}
    \end{proof}

    \begin{proof}[Proof of \Cref{lem:sum2}]
    Here, we upper bound each $\Prs{\cE^{ij}_k}$ individually and union bound over the at most $q$ choices of $k$ and $n$ choices of $j$. Fix agent $j\neq i$ and $k$ such that  $32\log m + 4n \le k \le q - 32\log m - 5n$. We first consider  $j < i$ and  later show how to extend it to $j > i$. 
    Notice that $b_{ik} \succeq_i g_{(k - 1)n + i}$, i.e., $i$'s $k$'th favorite item is at least as good as their $k$'th pick in the Round Robin phase (though it may be better if $i$ was lucky during the giveaway phase). Additionally, if $b_{jk} \succ_i b_{ik} \succeq_i g_{(k - 1)n + i}$, this implies that $A_j$ contains at least $k$ items strictly preferred to $g_{(k - 1)n + i}$. None of these could have been chosen after time $(k - 1)n + i$ in the Round Robin phase as $g_{(k - 1)n + i}$ was the best available item from $P_{(k - 1)n + i - 1}$. Additionally, $g_{(k - 1)n + i} \succ \gamma_{ij}$ (the item $i$ gave to $j$) as  $\gamma_{ij}$ was the worst available item at a time when $g_{(k - 1)n + i}$ was available. Hence, for $\cE^{ij}_k$ to hold, at least $k$ of $\set{\gamma_{i'j} \suchthat i' \ne i, j} \cup \set{g_j, \ldots, g_{(k - 1)n + j}}$ are preferred to $g_{(k - 1)n + i}$. Since $|\set{\gamma_{i'j} \suchthat i' \ne i, j}| = n - 2$, this implies that $|\set{\ell \suchthat g_{(\ell - 1)n + j} \succ_i g_{(k - 1)n + i}}| \ge k - (n - 2)$. More formally, we have that
    \[
        \Prs{\cE^{ij}_k} \le \Prs{\sum_{\ell = 1}^k \mathbb{I}[g_{(\ell - 1)n + j} \succ_i g_{(k - 1)n + i}] \ge k - (n - 2)}
    \]
    
    As before, a necessary condition for $g_{(\ell - 1)n + j} \succ_i g_{(k - 1)n + i}$ is $g_{(\ell - 1)n + j} \succeq_i P^{((k - \ell)n + i - j)}_{(\ell - 1)n + j -1}$ , hence,
    \begin{align*}
        \Prs{\sum_{\ell = 1}^k \mathbb{I}[g_{(\ell - 1)n + j} \succ_i g_{(k - 1)n + i}] \ge k - (n - 2)}
        &\le  \Prs{\sum_{\ell = 1}^k \mathbb{I}[g_{(\ell - 1)n + j} \succeq_i P^{((k - \ell)n + i - j)}_{(\ell - 1)n + j - 1}] \ge k - (n - 2)}\\
        &= \Prs{\sum_{\ell = 1}^k \mathbb{I}[g_{(\ell - 1)n + j} \nsucceq_i P^{((k - \ell)n + i - j)}_{(\ell - 1)n + j - 1}] \le n - 2}.
    \end{align*}
    Now each of these indicators is an independent Bernoulli random variable. This means we can use a  Chernoff bound on the sum. To do so, we first bound the expectation of the  sum. We have that 
    \begin{align*}
        \E\left[\sum_{\ell = 1}^k \mathbb{I}[g_{(\ell - 1)n + j} \nsucceq_i P^{((k - \ell)n + i - j)}_{(\ell - 1)n + j - 1}]\right]
        &= \sum_{\ell = 1}^k\Prs{g_{(\ell - 1)n + j} \nsucceq_i P^{((k - \ell)n + i - j)}_{(\ell - 1)n + j - 1}}\\
        &= \sum_{\ell = 1}^k 1 - \frac{(k - \ell)n + i - j}{|P_{(\ell - 1)n + j - 1}|}\\
        &= \sum_{\ell = 1}^k 1 - \frac{(k - \ell)n + i - j}{n(q - (n - 1)) - ((\ell - 1)n + j - 1)}\\
        &= \sum_{\ell = 1}^k 1 - \frac{(k - \ell)n + i - j}{n(q - n -  \ell) + 1 - j}\\
        &= \sum_{\ell = 1}^k \frac{n(q - n - k)  + 1 - i}{n(q - n -  \ell) + 1 - j}\\
        &\ge \sum_{\ell = 1}^k \frac{n(q - n - k)  + 1 - i + (i - 1)}{n(q - n -  \ell) + 1 - j + (i - 1)}\\
        &= \sum_{\ell = 1}^k \frac{n(q - n - k)}{n(q - n -  \ell) +i - j}\\
        &\ge \sum_{\ell = 1}^k \frac{n(q - n - k)}{n(q - n -  \ell) + n}\\
        &\ge \sum_{\ell = 1}^k \frac{q - n - k}{q - n -  \ell + 1}
        \\
        &\ge  k\cdot \frac{q - n - k}{q - n + 1}\\
        &\ge  \frac{(q - n - k)k}{q - n + 1} + \frac{ (q - n - k)}{q - n + 1} - 1\\
        &=  \frac{(q - n - k)(k + 1)}{q - n + 1} - 1\\
        &\ge \frac{\max(q - n - k, k + 1)\cdot \min(q - n - k, k + 1)}{q - n + 1} - 1\\
        &\ge \frac{1}{2} \cdot \min(q - n - k, k + 1) - 1 
        \\
        &\ge 16\log m + 2n - 1.
    \end{align*}
    Let $\mu = \E\left[\sum_{\ell = 1}^k \mathbb{I}[g_{(\ell - 1)n + j} \nsucceq_i P^{((k - \ell)n + i - j)}_{(\ell - 1)n + j - 1}]\right]$. The above implies that $(n - 2) \le \frac{1}{2}\mu$. Hence, a Chernoff bound implies that
    \begin{align*}
       \Prs{\sum_{\ell = 1}^k \mathbb{I}[g_{(\ell - 1)n + j} \nsucceq_i P^{((k - \ell)n + i - j)}_{(\ell - 1)n + j - 1}] \le n - 2}
        &\le \Prs{\sum_{\ell = 1}^k \mathbb{I}[g_{(\ell - 1)n + j} \nsucceq_i P^{((k - \ell)n + i - j)}_{(\ell - 1)n + j - 1}] \le \frac{1}{2}\mu}\\
        &\le \exp\left(-\frac{\left(\frac{1}{2}\right)^2 \mu }{2} \right)\\
        &= \exp\left(-\frac{ \mu }{8} \right)\\
        &\le \exp\left(-2\log m \right) = \frac{1}{m^2}.
    \end{align*}
    All together, this implies that for $j < i$, $\Prs{\cE^{ij}_k} \le 1/m^2$.
    
    For $j > i$, the argument can be modified as follows. Notice that  $j$'s $k$'th pick will be after $i$'s $k$'th pick, hence, we need only sum over $\ell$ from $1$ to $k - 1$. This gives 
    \begin{align*}
        \Prs{\sum_{\ell = 1}^{k-1} \mathbb{I}[g_{(\ell - 1)n + j} \succ_i g_{(k - 1)n + i}] \ge k - (n - 2)}
        &\le  \Prs{\sum_{\ell = 1}^{k-1} \mathbb{I}[g_{(\ell - 1)n + j} \succeq_i P^{((k - \ell)n + i - j)}_{(\ell - 1)n + j - 1}] \ge k - (n - 2)}\\
        &= \Prs{\sum_{\ell = 1}^{k-1} \mathbb{I}[g_{(\ell - 1)n + j} \nsucceq_i P^{((k - \ell)n + i - j)}_{(\ell - 1)n + j - 1}] \le n - 3}.
    \end{align*}
    
    The analysis of the expectation goes through (replacing $k$ with $k - 1$) since we never use that $j < i$, at least until getting the expectation is lower bounded by 
    \[
        \sum_{\ell = 1}^{k-1} \frac{q - n - k}{q - n - \ell + 1}.
    \]
    Now, we  can continue with 
    \begin{align*}
        \sum_{\ell = 1}^{k-1} \frac{q - n - k}{q - n - \ell + 1} \geq \sum_{\ell - 1}^{k} \frac{q - n - k}{q - n - \ell + 1} - 1, 
    \end{align*}
    at which point the remainder of the analysis continues to hold (off by one) until  $\mu \ge 16\log m + 2n - 2$. 
    This  is still sufficient for $(n - 3) \le \frac{1}{2}\mu$ and $\mu \ge 16\log m$, which is all that was needed for the Chernoff bound to hold. This means that $\Prs{\cE^{ij}_k} \le 1/m^2$ for $j > i$ as well.
    
    Union bounding over the $q$ possible choices of $k$ and $n$ possible choices of $j$, we get that
    \begin{equation*}
        \Prs{\bigcup_{k: 32\log m + 4n \le k \le q - 32\log m - 5n} \bigcup_{j: j \ne i} \cE^{ij}_k} \le \frac{qn}{m^2} = \frac{1}{m}. \qedhere
    \end{equation*}
    \end{proof}
    
    \begin{proof}[Proof of \Cref{lem:sum3}]
    Let $L$ be the last $\floor{64n \log m + 12n^2}$ items according to $\pi_i$. Just as in Lemma~\ref{lem: bound for k = q, EF}, we will decompose
    \begin{align*}
        \Prs{\bigcup_{k: q - 32\log m - 5n <  k \le q - 1 } \bigcup_{j: j \ne i} \cE^{ij}_k}
        &\le \Prs{\bigcup_{k: q - 32\log m - 5n <  k \le q - 1 } \bigcup_{j: j \ne i} \cE^{ij}_k \suchthat b_{(q-1)i} \notin L} + \Prs{b_{(q-1)i} \in L}\\
        &\le \sum_{j: j \ne i} \Prs{\bigcup_{k: q - 32\log m - 5n <  k \le q - 1 } \cE^{ij}_k \suchthat b_{(q-1)i} \notin L} + \Prs{b_{(q-1)i} \in L}.
    \end{align*}
    We upperbound each of these summands. We begin with $\Prs{b_{(q-1)i} \in L}$. We again slit decompose this depending on if $i$ received any goods in $L$ during the giveaway phase.
    \[
        \Prs{b_{(q-1)i} \in L} \le \Prs{b_{(q-1)i} \suchthat \set{\gamma_{ji} \suchthat j \ne i} \cap L = \emptyset} + \Prs{\set{\gamma_{ji} \suchthat j \ne i} \cap L \ne \emptyset}.
    \]
    For $\Prs{\set{\gamma_{ji} \suchthat j \ne i} \cap L \ne \emptyset}$, recall that $\gamma_{ji}$ are just random items from the pool from when the pool is of size at least $m - n^2$, hence, the probability each is an element of $L$ is at most $\frac{|L|}{m - n^2}$. Using the assumption that $m \ge 2n^2$, we have that this is at most $\frac{2|L|}{m} \le \frac{128n\log m + 20 n^2}{m}$.
    
    For $\Prs{b_{(q-1)i} \in L \suchthat \set{\gamma_{ji} \suchthat j \ne i} \cap L = \emptyset}$, since we are conditioning on $i$ not being given any items in $L$ during the giveaway phase, the only way that $i$ will receive such items is if they pick them during round robin. For $b_{(q-1)i} \in L$, $i$ must receive at least $2$ items in $L$. This means that even when $i$ made their second to last pick at Round Robin $g_{n(q - n - 1) + i}$ the pool was  $P_{n(q - n - 1) + i - 1}$, the only available items were from $L$, i.e.,$P_{n(q - n - 1) + i - 1} \subseteq L$. Notice that $|P_{n(q - n - 1) + i - 1}| = m - n(n - 1) - (n(q - n - 1) + i - 1) = 2n - i + 1$. We will union bound over all possible sets $S \in \binom{L}{2n - i + 1}$ analyzing the probability that $\Prs{P_{n(q - n - 1) + i - 1} = S \suchthat \set{\gamma_{ji} \suchthat j \ne i} \cap L \ne \emptyset}$. 
    
    We can decompose
    \begin{align*}
       &\Prs{P_{n(q - n - 1) + i - 1} = S \suchthat \set{\gamma_{ji} \suchthat j \ne i} \cap L \ne \emptyset}\\
        &\qquad= \Prs{S \subseteq P_0 \suchthat \set{\gamma_{ji} \suchthat j \ne i} \cap L \ne \emptyset} \cdot \prod_{t = 1}^{P_{n(q - n - 1) + i - 1}} \Prs{g_t \notin S \suchthat S \subseteq P_{t - 1} \land \set{\gamma_{ji} \suchthat j \ne i} \cap L \ne \emptyset}.
    \end{align*}
    For the first term and all terms where $t \equiv i \bmod n$, this probability may be hard to compute, but we can upperbound it by $1$. For all other terms with $t \not\equiv i \bmod n$,
    \[
        \Prs{g_t \notin S \suchthat S \subseteq P_{t - 1} \land \set{\gamma_{ji} \suchthat j \ne i} \cap L \ne \emptyset} = 1 - \frac{|S|}{|P_{t - 1|}}.
    \]
    Hence, this entire product is upper bounded by
    \[
        \prod_{t = 1: t \not\equiv i \bmod n}^{n(q - n - 1) + i - 1} \frac{n(q - (n - 1)) + 1 - t - (2n - i + 1)}{n(q - (n - 1) + 1 - t}.
    \]
    Notice that this is now the same analysis as the proof of \Cref{claim:bound on L} with $m$ replaced with $q(n - 1)$ and $n + 1 -i$ replaced with $2n + 1 - i$. Hence, the above product is upperbounded by
    \[
        2\left( \frac{2}{(n(q - (n - 1)))^{1 - 1/n}}\right)^{2n + 1 - i}.
    \]
    Using the fact that $m \ge 2n^2$ so $n(q - (n - 1)) \ge m/2$, this is at most
    \[
        2\left( \frac{4}{m^{1 - 1/n}}\right)^{2n + 1 - i}.
    \]
    Union bounding over the at most $\binom{|L|}{2n + 1 - i} \le |L|^{2n + 1 - i} \le (64n \log m + 10n^2)^{2n + 1 - i}$ subsets $S$, we get that
    \[
     \Prs{b_{(q-1)i} \suchthat \set{\gamma_{ji} \suchthat j \ne i} \cap L = \emptyset} \le \left( \frac{256n \log m + 48n^2}{m^{1 - 1/n}} \right)^{2n + 1 - i} \le \left( \frac{256n \log m + 48n^2}{m^{1/2}} \right)^{2n + 1 - i}.
    \]
    Since we have restricted to $m \ge (256 n \log m + 40n^2)^2$, $m^{1/2} \ge 256 n \log m + 48 n^2$, so the inside is at most $1$. Hence, since 
    $2n + 1 - i \ge n + 1 \ge 2$, this is at most
    \[
        \left( \frac{64n\log m + 12n^2}{m^{1/2}} \right)^{2} = \frac{65536 n^2\log^2 m + 24576 n^3 \log m + 2304 n^4}{m}.
    \]
    Hence, we have that
    \[
        \Prs{b_{(q-1)i} \in L} \le \frac{65536 n^2\log^2 m + 24576 n^3 \log m + 2304 n^4 + 128n \log m + 24n^2}{m}.
    \]
    
    Finally, fix some agent $j \ne i$. We consider $\Prs{\bigcup_{k: q - 32\log m - 5n <  k \le q - 1 } \cE^{ij}_k \suchthat b_{(q-1)i} \notin L}$. Notice that if $b_{(q - 1)i} \notin L$, as long as $|A_j \cap L| \ge 32\log m + 5n$, then for all $k \le 32 \log m + 5n$, $b_{ik} \succeq_i b_{i(q - 1)} \succ_i b_{jk}$, so $\cE^{ij}_k$ does not hold. Hence, we wish to upperbound the probability that $|A_j \cap L| < 32 \log m + 5n$. We use the second equivalent process, and consider sampling all of $\gamma_{i'j'}$ items along with sets $R_0, \ldots, R_q$ and goods picked by agent $i$, $g_t$ with $t \equiv i \bmod n$, but we have not yet sampled which non-$i$ agent receives which items in $R_0, \ldots, R_q$. Notice that conditioned on $b_{i(q - 1)} \notin L$, $|A_i \cap L| \le 1$. Since $n(n - 1)$ items are in the giveaway phase and $|R_0| + |R_q| = n - 1$, this implies that $\left| L \cap \left(\bigcup_{\ell = 1}^{q - 1} R_\ell \right) \right| \ge |L| - n^2 \ge \floor{64n \log m + 11n^2} \ge 64n\log m + 10 n^2$. Notice that $|A_j \cap L| \ge \sum_{\ell = 1}^{q - 1} \mathbb{I}[A_j \cap R_l \cap L \ne \emptyset]$. Further, this is a sum of independent Bernoulli variables with $\Prs{A_j \cap R_l \cap L \ne \emptyset} = \frac{|R_l \cap L|}{n - 1}$. Hence, $\mathbb{E}[\sum_{\ell = 1}^{q - 1} \mathbb{I}[A_j \cap R_l \cap L \ne \emptyset]] \ge \sum_{\ell = 1}^{q - 1} \frac{|R_l \cap L|}{n - 1} = \frac{\left| L \cap \left(\bigcup_{\ell = 1}^{q - 1} R_\ell \right) \right|}{n - 1} \ge 64 \log m + 10n$. We can then use a Chernoff bound to show that
    \begin{align*}
        \Prs{|A_j \cap L| \le 32\log m + 5n}
        &\le \Prs{\sum_{\ell = 1}^{q - 1} \mathbb{I}[A_j \cap R_l \cap L \ne \emptyset] \le 32\log m + 5n}\\
        &\le \Prs{\sum_{\ell = 1}^{q - 1} \mathbb{I}[A_j \cap R_l \cap L \ne \emptyset] \le \frac{1}{2}\mathbb{E}\left[\sum_{\ell = 1}^{q - 1} \mathbb{I}[A_j \cap R_l \cap L \ne \emptyset]\right]}\\
        &\le \exp\left(-\frac{(1/2)^2 \mathbb{E}\left[\sum_{\ell = 1}^{q - 1} \mathbb{I}[A_j \cap R_l \cap L \ne \emptyset]\right]}{2} \right)\\
        &\le \exp\left(-\frac{32 \log m}{8} \right) \le \frac{1}{m^4} \le \frac{1}{m}.
    \end{align*}
    Hence, we have that
    \[
        \sum_{j: j \ne i} \Prs{\bigcup_{k: q - 32\log m - 5n <  k \le q - 1 } \cE^{ij}_k \suchthat b_{(q-1)i} \notin L} \le \frac{n}{m}.
    \]
    Putting this together, we have that
    \begin{align*}
        &\Prs{\bigcup_{k: q - 32\log m - 5n <  k \le q - 1 } \bigcup_{j: j \ne i} \cE^{ij}_k}\\
        &\qquad \le  \frac{65536 n^2\log^2 m + 24576 n^3 \log m + 2304 n^4 + 128n \log m + 24n^2 + n}{m}. \qedhere
    \end{align*}
    \end{proof}

    \begin{proof}[Proof of \Cref{lem:sum4}]
    Let $L$ be the last $n - 1$ items according to $\pi_i$. A sufficient condition for $\bigcup_{k: k = q} \bigcup_{j: j \ne i} \cE^{ij}_k = \bigcup_{j: j \ne i} \cE^{ij}_q$ to hold is that $\set{\gamma_{ij} \suchthat i \ne j} = L$, i.e., $i$ is able to give away their bottom $n - 1$ goods. In this case, $b_{iq} \succ_i L$ while $b_{jq} \in L$ for all $j \ne i$. For this to hold, the only necessary condition is that no other agent has given away a good in $L$ prior to $i$ giving them away. There are at most $n^2$ items given away before $i$ is able to give, and each of these is an element of $L$ with probability at most $\frac{|L|}{m - n^2}$. Since we are assuming that $M \ge 2n^2$, this is at most $\frac{2n}{m}$. Hence,
    \begin{equation*}
        \Prs{\bigcup_{k: k = q} \bigcup_{j: j \ne i} \cE^{ij}_k} \le \frac{2n}{m}.\qedhere
    \end{equation*}
\end{proof}

\section{Non-existence of $\SDEFX$ allocations}\label{appendix:no-sdefx}
Complementing Theorem \ref{thm:positive for order consistent}, in this section we show that $\SDEFX$ allocations may not exist, even for additive valuations. Additionally, this also contrasts with the fact that the existence of $\EFX$ remains an elusive open problem.

We first consider an instance $m=4$ items and $n=2$ agents having identical additive valuation function $v: 2^{[4]} \to \mathbb{R}_{\geq 0}$. The function $v$ is such that $v(\{1\}) = 4$, $v(\{2\}) = 1+\epsilon$, $v(\{3\}) = 1$, and $v(\{4\}) = 1 - \epsilon$. Indeed, the only $\EFX$ allocations in this instance are $\alloc = (\set{1}, \set{2,3,4})$ and $\widehat{\alloc} = (\set{2,3,4}, \set{1})$. Note that an $\SDEFX$ allocation must also be $\EFX$, therefore, if an $\SDEFX$ allocation exists for this instance it must be $\alloc$ or $\widehat{\alloc}$.

Furthermore, if an allocation is $\SDEFX$, then it must also be $\EFX$ for any other instance in which both agents have the same preference order over the items. Now consider a second instance wherein both agents instead have the additive function $\widehat{v}$ such that $\widehat{v}(\{1\}) = 4$, $\widehat{v}(\{2\}) = 3$, $\widehat{v}(\{3\}) = 2$, and $\widehat{v}(\{4\}) = 1$. Note that the preference order of items is the same in $v$ and $\widehat{v}$, therefore, if $\alloc$ (or $\widehat{\alloc}$) is $\SDEFX$ then it must be $\EFX$ for the second instance. It is easy to see that neither $\alloc$ nor $\widehat{\alloc}$ is $\EFX$ for the second instance. This shows that $\SDEFX$ allocations do not exist, even for additive valuations.

\end{document}